\pdfoutput=1
\RequirePackage{fix-cm}
\documentclass[smallextended]{svjour3}       
\smartqed  
\usepackage{hyperref}
\usepackage{array}
\usepackage{rotating}
\newcolumntype{H}{>{\setbox0=\hbox\bgroup}c<{\egroup}@{}}
\usepackage{graphicx}
\usepackage{amssymb}
\usepackage{amsmath}
\usepackage{pgf,tikz}
\usepackage{mathrsfs}
\usetikzlibrary{arrows}
\usepackage{pgfplots}
\usepackage[colorinlistoftodos,prependcaption,textsize=tiny]{todonotes}
\usepackage[numbers]{natbib}
\usepackage{algorithm}
\usepackage{algorithmicx}
\usepackage{algpseudocode}
\usepackage{mathabx}
\usepackage{multirow}
\newcommand{\sequal}{\!=\!}

\newcommand{\sminus}{\!\!-\!\!}
\newcommand{\splus}{\!\!+\!\!}

\DeclareMathOperator*{\argmax}{argmax}

\DeclareMathOperator*{\bd}{bd}
\DeclareMathOperator*{\aff}{aff}
\DeclareMathOperator*{\lexmax}{lex--max}

\allowdisplaybreaks
\def\makeheadbox{\relax}

\begin{document}

\title{Computing a Pessimistic Leader-Follower Equilibrium with Multiple Followers: the Mixed-Pure Case}

\titlerunning{Computing Pessimistic Leader-Follower Equilibria: the Mixed-Pure Case}        

\author{Stefano Coniglio \and Nicola Gatti \and Alberto Marchesi}

\authorrunning{S. Coniglio \and N. Gatti \and A. Marchesi} 

\institute{Stefano Coniglio \at
              University of Southampton, University Road, Southampton, SO17 1BJ, UK\\
              Tel.: +44-023-8059-4546\\
              \email{s.coniglio@soton.ac.uk}           
           \and
           Nicola Gatti \at
              Politecnico di Milano, piazza Leonardo da Vinci 32, Milano, 20133, Italy \\
              Tel.: +39-02-2399-3658\\
              Fax.: +39-02-2399-3411\\
   	    \email{nicola.gatti@polimi.it}
	    \and
           Alberto Marchesi \at
              Politecnico di Milano, piazza Leonardo da Vinci 32, Milano, 20133, Italy \\
              Tel.: +39-02-2399-9685\\
              Fax.: +39-02-2399-3411\\
   	    \email{alberto.marchesi@polimi.it}
}

\date{}

\maketitle




\begin{abstract}
The search problem of computing a {\em leader-follower equilibrium} (also referred to as an \emph{optimal strategy to commit to}) has been widely investigated in the scientific literature in, almost exclusively, the single-follower setting. Although the \emph{optimistic} and \emph{pessimistic} versions of the problem, i.e., those where the single follower breaks any ties among multiple equilibria either in favour or against the leader, are solved with different methodologies, both cases allow for efficient, polynomial-time algorithms based on linear programming. The situation is different with multiple followers, where results are only sporadic and depend strictly on the nature of the followers' game.

In this paper, we investigate  the setting of a normal-form game with a single leader and multiple followers who, after observing the leader's commitment, play a Nash equilibrium. The corresponding search problem, both in the optimistic and pessimistic versions, is known to be not in Poly-$\textsf{APX}$ unless $\textsf{P}=\textsf{NP}$ and exact algorithms are known only for the optimistic case.
We focus on the case where the followers play in pure strategies---a restriction that applies to a number of real-world scenarios and which, in principle, makes the problem easier---under the assumption of pessimism (as it is easy to show, the optimistic version of the problem can be straightforwardly solved in polynomial time). After casting this search problem as a {\em pessimistic bilevel programming problem}, we show that, with two followers, the problem is \textsf{NP}-hard and, with three or more followers, it is not in Poly-\textsf{APX} unless $\textsf{P}=\textsf{NP}$. This last result matches the inapproximability result which holds for the unrestricted case and shows that, differently from what happens in the optimistic version, hardness in the pessimistic problem is not due to the adoption of mixed strategies.
We then show that the problem admits, in the general case, a supremum but not a maximum, and we propose a single-level mathematical programming reformulation which calls for the maximisation of a nonconcave quadratic function over an unbounded nonconvex feasible region defined by linear and quadratic constraints.
Since, due to admitting a supremum but not a maximum, only a restricted version of this formulation can be solved to optimality with state-of-the-art methods, we propose an exact {\em ad hoc} algorithm, which we also embed within a branch-and-bound scheme, capable of computing the supremum of the problem and, for cases where there is no leader's strategy where such value is attained, also an $\alpha$-approximate strategy where $\alpha > 0$ is an arbitrary additive loss.
We conclude the paper by evaluating the scalability of our algorithms via computational experiments on a well-established testbed of game instances.
\end{abstract}


\section{Introduction}\label{sec:introduction}

In recent years, {\em Leader-Follower} (or {\em Stackelberg}) {\em Games} (LFGs) and their corresponding {\em Leader-Follower Equilibria} (LFEs) have attracted a growing interest in many disciplines, including theoretical computer science, artificial intelligence, and operations research. LFGs describe situations where one player (the {\em leader}) commits to a {\em strategy} and the other players (the {\em followers}) first observe the leader's commitment and, then, decide how to play. In the literature, LFEs are often referred to as \emph{optimal strategies} (for the leader) {\em to commit to}. 
LFGs encompass a broad array of real-world games. A prominent example is that one of security games, where a defender, acting as a leader, is tasked to allocate scarce resources to protect valuable targets from an attacker, acting as a follower~\cite{an2011guards,KiekintveldJTPOT09,paruchuri2008playing}. Besides the security domain, applications can be found in, among others, interdiction games~\cite{caprara2016bilevel,matuschke2017protection}, toll-setting problems~\cite{labbe2016bilevel}, and network routing~\cite{amaldi2013network}.

While, to the best of our knowledge, the majority of the game theoretical investigations on the computation of LFEs assumes the presence of a single follower, we address, in this work, the multi-follower case.

When facing an LFG and, in particular, a multi-follower one, two aspects need to be considered: the \emph{type} of game (induced by the leader's strategy) the followers play and, in it, \emph{how} ties among the multiple equilibria which could arise are broken.

As to the nature of the followers' game, and restricting ourselves to the cases which look more natural, the followers may play hierarchically one at a time, as in a hierarchical Stackelberg game~\cite{Conitzer:2006}, simultaneously and cooperatively~\cite{ConitzerK11}, or simultaneously and noncooperatively~\cite{basilico2016methods}.

As to breaking ties among multiple equilibria, it is natural to consider two cases: the {\em optimistic} one, where the followers end up playing an equilibrium which {\em maximises} the leader's utility, and the {\em pessimistic} one, where they end up playing an equilibrium by which the leader's utility is {\em minimised}. Note that we are not assuming, here, that the followers could {\em agree} on an optimistic or pessimistic equilibrium in a practical application. Rather, the optimistic and pessimistic cases allow for the computation of the tightest range of values the leader's utility may take without making any assumptions on which equilibrium the followers would {\em actually} end up playing. From this perspective, while an optimistic LFE accounts for the best case for the leader, a pessimistic LFE accounts for the worst case.
%
In this sense, the computation of a pessimistic LFE is paramount in realistic scenarios as, differently from an optimistic one, the former is robust. As we will see, though, this degree of robustness comes at a high computational price, as computing a pessimistic LFE is a much harder task than computing its optimistic counterpart.

\subsection{Leader-Follower Nash Equilibria}

Throughout the paper, we will focus on the case of normal-form games where, after the leader's commitment to a strategy, the followers play {\em simultaneously} and {\em noncooperatively}, reaching a Nash equilibrium. We refer to the corresponding equilibrium as {\em Leader-Follower Nash Equilibrium} (LFNE).

In particular, we consider the case where the followers are restricted to pure strategies.
This restriction is motivated by some reasons. First, the problem of finding an NE in mixed strategies, a subproblem of that of finding an LFNE, is already hard with two or more players (which clearly implies the hardness of finding an LFNE in mixed strategies)---differently from the problem of computing an NE in pure strategies, which can be solved in polynomial time. Thus, the study of the case where followers play pure strategies does not appear direct and could allow one to characterise more accurately the tractability of the problem. Secondly, many games admit pure-strategy NEs, among which potential games~\cite{monderer1996potential}, congestion games~\cite{rosenthal1973class}, and toll-setting problems~\cite{labbe2016bilevel}. The same also holds, with high probability, in many unstructured games (see Subsection~\ref{subsec:preliminary}).

\subsection{Original Contributions}

After briefly pointing out that an optimistic LFNE (with followers restricted to pure strategies) can be computed efficiently (in polynomial time) by a mixture of enumeration and linear programming, we entirerly devote the remainder of the paper to the pessimistic case (with, again, followers restricted to pure strategies). In terms of computational complexity, we show that, differently from the optimistic case, in the pessimistic case the equilibrium-finding problem 
is \textsf{NP}-hard with two or more followers and not in Poly-\textsf{APX} when the number of followers is three or more unless $\textsf{P}=\textsf{NP}$. To establish these two results, we introduce two reductions, one from Independent Set and the other one from 3-SAT.

After analysing the complexity of the problem, we focus on its algorithmic aspects.
First, we formulate the problem as a {\em pessimistic bilevel programming problem with multiple followers}. We, then, show how to recast it as a single-level Quadratically Constrained Quadratic Program (QCQP), which we show to be impractical to solve due to admitting a supremum, but not a maximum. We, then, introduce a restriction based on a Mixed-Integer Linear Program (MILP) which, while forsaking optimality, always admits an optimal (restricted) solution.
Next, we propose an exact algorithm to compute the value of the supremum of the problem, based on an enumeration scheme which, at each iteration, solves a lexicographic MILP (lex-MILP) where the two objective functions are optimised in sequence. Subsequently, we embed the enumerative algorithm within a branch-and-bound scheme---obtaining an algorithm which is, in practice, much faster. We also extend the algorithm (in both versions) so that, for cases where the supremum is not a maximum, it returns a strategy by which the leader can obtain a utility within an additive loss $\alpha$ with respect to the supremum, for an arbitrarily chosen $\alpha > 0$.
%
To conclude, we experimentally evaluate the scalability of our methods over a rich testbed of instances which is standard in game theory.

The status, in terms of complexity and known algorithms, of the problem of computing an LFNE (with followers playing pure or mixed strategies) is summarised in Table~\ref{tab:summaryofresults}. The original results we provide in this paper are reported in boldface. 

\begin{table}[t]
\caption{Summary of known results for the computation of an LFNE. The entries in boldface correspond to original contributions of this work. The number of players is denoted by $n$.}
\label{tab:summaryofresults}
\hspace{-0.1cm}\includegraphics[width=\textwidth
]{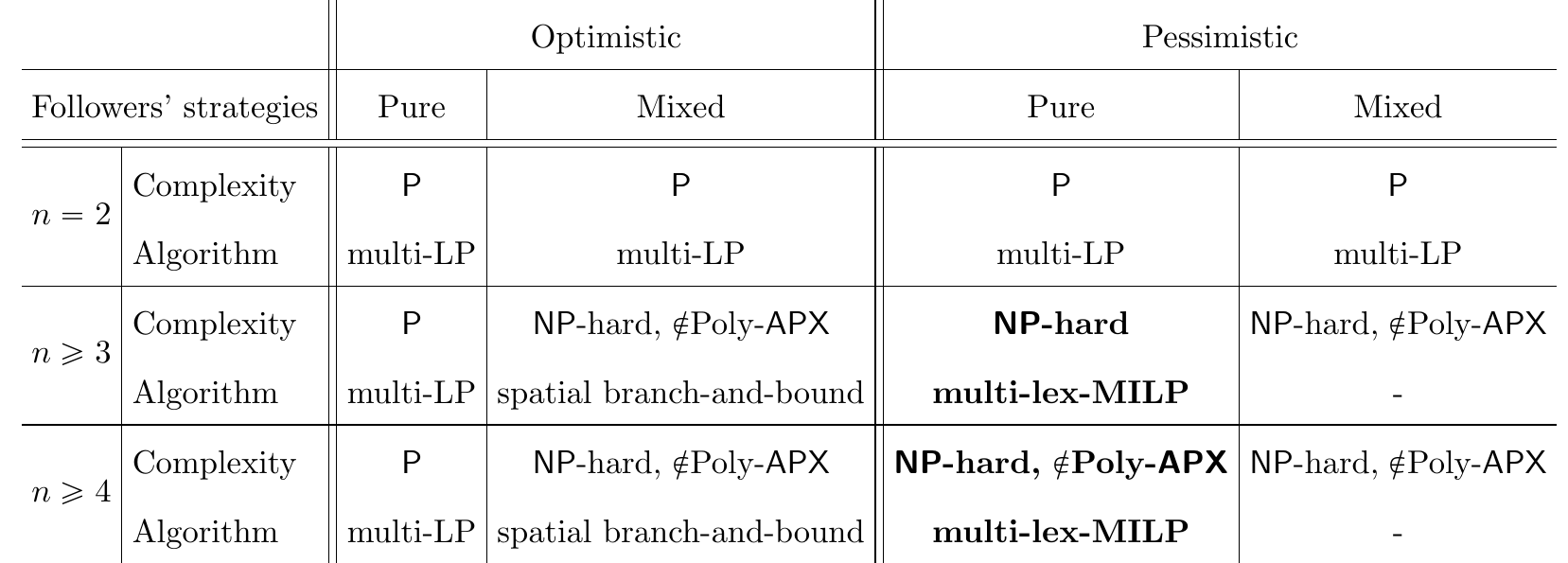}
\end{table}

\subsection{Paper Outline}

The paper is organised as follows.\footnote{A preliminary version of this work appeared in~\cite{coniglio2017pessimistic}.} Previous works are introduced in Section~\ref{sec:previouswork}. The problem we study is formally stated in Section~\ref{sec:problem}, together with some preliminary results. In Section~\ref{sec:reduction}, we present the computational complexity results. Section~\ref{sec:reform} introduces the single-level reformulation(s) of the problem, while Section~\ref{sec:exact_algorithm} describes our exact algorithm (in its two versions). An empirical evaluation of our methods is carried out in Section~\ref{sec:experiments}. Section~\ref{sec:conclusions} concludes the paper.

\section{Previous Works}\label{sec:previouswork}

As we mentioned in Section~\ref{sec:introduction}, most of the works on (normal-form) LFGs focus on the single-follower case. In such case, as shown in~\cite{Conitzer:2006}, the follower always plays a pure strategy, except for degenerate games. In the optimistic case, an LFE can be found in polynomial time by solving a Linear Program (LP) for each action of the (single) follower (the algorithm is, thus, a multi-LP). Each LP maximises the expected utility of the leader, subject to a set of constraints imposing that the given follower's action is a best-response~\cite{Conitzer:2006}. As shown in~\cite{ConitzerK11}, all these LPs can be encoded into a single LP---a slight variation of the LP that is used to compute a \emph{correlated equilibrium} (the solution concept where all the players can exploit a correlation device to coordinate their strategies).\footnote{In this case, the leader and the follower play correlated strategies under the rationality constraints imposed on the follower only, maximising the leader's expected utility.} Some works study the equilibrium-finding problem (only in the optimistic version) in structured games where the action space is combinatorial. See~\cite{DBLP:journal/ai/BasilicoNG17} for more references.

For what concerns the pessimistic single-follower case, the authors of~\cite{leaderfollower} study the problem of computing the supremum of the leader's expected utility.
They show that, for the latter, it suffices to consider the follower's actions which constitute a best-response to a full-dimensional region of the leader's strategy space. The multi-LP algorithm the authors propose solves two LPs per action of the follower, one to verify whether the best-response region for that action is full-dimensional (so to discard it if full-dimensionality does not hold) and a second one to compute the best leader's strategy within that best-response region. The algorithm runs in polynomial time. 
%
While the authors limit their analysis to computing the supremum of the leader's utility, we remark that such value does not always translate into a strategy that the leader can play as, in the general case where the leader's utility does not admit a maximum, there is no leader's strategy giving her a utility equal to the supremum. In such cases, one should rather look for a strategy providing the leader with an expected utility which approximates the value of the supremum. This aspect, which is not addressed in~\cite{leaderfollower}, will be tackled, on the multi-follower case, by our work.

The multi-follower case, which, to the best of our knowledge, has only been investigated in~\cite{basilico2016methods,basilico2017bilevel}, is computationally much harder than the single-follower case, being, in the general case with the leader and the followers entitled to mixed strategies, \textsf{NP}-hard and inapproximable, in polynomial time, to within any polynomial factors unless $\mathsf{P}=\mathsf{NP}$. In the aforementioned works, the problem of finding an equilibrium in the optimistic case is formulated as a nonlinear and nonconvex mathematical program and solved to global optimality (within a given tolerance) with spatial branch-and-bound techniques. No exact methods are proposed for the pessimistic case.

\section{Problem Statement and Preliminary Results}\label{sec:problem}

After setting the notation used throughout the paper, this section offers a formal definition of the equilibrium-finding problem we tackle in this work and illustrates some of its properties.

\subsection{Notation}

Let $N=\{1,\dots,n\}$ be the set of players and, for each player $p \in N$, let $A_p$ be her set of actions, of cardinality $m_p = |A_p|$. Let also $A = \bigtimes_{p \in N} A_p = A_1 \times \dots \times A_n$. For each player $p \in N$, let $x_p \in [0,1]^{m_p}$, with $\sum_{a_p \in A_p}x_p^{a_p} = 1$, be her {\em strategy vector} (or strategy, for short), where each component $x_{p}^{a_p}$ of $x_p$ represents the probability by which player~$p$ plays action~$a_p \in A_{p}$. For each player $p \in N$, let also $\Delta_p = \{ x_p \in [0,1]^{m_p} : \sum_{a_p \in A_p} x_p^{a_p} = 1 \}$ be the set of her strategies, or {\em strategy space}, which corresponds to the standard $(m_p-1)$-simplex in $\mathbb{R}^{m_p}$. A strategy is said {\em pure} when only one action is played with positive probability, i.e., when $x_p \in \{0,1\}^{m_p}$, and {\em mixed} otherwise. In the following, we denote the collection of strategies of the different players, or {\em strategy profile}, by $x=(x_{1}, \ldots, x_{n})$. For the case where all strategies are pure, we denote the collection of actions played by the players, or {\em action profile}, by $a = (a_1, \ldots, a_n)$.

Given a strategy profile $x$, we denote the collection of all the strategies in it but that one of player $p \in N$ by $x_{-p}$, i.e., $x_{-p}=(x_1,\ldots,x_{p-1},x_{p+1},\ldots,x_n)$. Given $x_{-p}$ and a strategy vector $x_p$, we denote the whole strategy profile $x$ by $(x_{-p},x_p)$. For action profiles, $a_{-p}$ and $(a_{-p},a_p)$ are defined analogously. For the case were all players are restricted to pure strategies, with the sole exception of player $p$, who is allowed to play mixed strategies, we use the notation $(a_{-p}, x_p)$.

We consider {\em normal-form games} where $U_p \in \mathbb{Q}^{m_1 \times \ldots \times m_n}$ represents, for each player $p \in N$, her (multidimensional) utility (or payoff) matrix, assuming, without loss of generality, entries in $[1, U_{max}]$ and $U_{max} \geq 1$. For each $p \in N$ and given an action profile $a=(a_1,\ldots,a_n)$, each component $U_p^{a_1 \ldots a_n}$ of $U_p$ corresponds to the utility of player~$p$ when all the players play the action profile~$a$. For the ease of presentation and when no ambiguity arises, we will often write, in the following,  $U_p^a$ in place of $U_p^{a_1 \ldots a_n}$ and, given a collection of actions $a_{-p}$ and an action $a_p \in A_p$, we will also use $U_p^{a_{-p},a_p}$ to denote the component of $U_p$ corresponding to the action profile $(a_{-p},a_p)$. Given a strategy profile $x=(x_1,\ldots,x_n)$, the expected utility of player~$p \in N$ is the $n$-th-degree polynomial $\sum_{a \in A} U_p^{a} x_1^{a_1} \, x_2^{a_2} \dots \, x_n^{a_n}$.

An action profile $a = (a_1, \ldots, a_n)$ is called {\em pure strategy Nash Equilibrium} (or pure NE, for short) if,
when the players in $N \setminus \{p\}$ play as the equilibrium prescribes, player $p$ cannot improve her utility by deviating from the equilibrium and playing another action $a_p' \neq a_p$, for all $p \in N$. More generally, a {\em mixed strategy Nash Equilibrium} (or mixed NE, for short) is a strategy profile $x = (x_1, \ldots ,x_n)$ such that no player $p \in N$ can improve her utility by playing a strategy $x_p' \neq x_p$, assuming the other players would play as the equilibrium prescribes. Observe that, in a normal-form game, a mixed NE always exists~\cite{Nash50}, while a pure NE may not. For more details on (noncooperative) game theory, we refer the reader to~\cite{shoham-book}.

Similar definitions hold for the case of LFGs when assuming that only a subset of players (the followers) play an NE, given the strategy the leader has committed to.

\subsection{The Problem and Its Formulation}\label{subsec:problem}

In the following, we assume that the $n$-th player takes the role of leader. We denote the set of followers (the first $n-1$ players) by $F=N \setminus \{n\}$. For the ease of notation, we also define  $A_F = \bigtimes_{p \in F} A_p$ as the set of followers' action profiles, i.e., the set of all collections of followers' actions. We also assume, unless otherwise stated, $m_p=m$ for every player $p \in N$, where $m$ denotes the number of actions available to each player. This is without loss of generality, as one can always introduce additional actions with a utility small enough to guarantee that they would never be played, so to obtain a game where each player has the same number of actions.

As we mentioned in Section~\ref{sec:introduction}, we tackle, in this work, the problem of computing an equilibrium in a normal-form game where the followers play a pure NE once observed the leader's commitment to a mixed strategy. We refer to an {\em Optimistic Leader-Follower Pure Nash Equilibrium} (O-LFPNE) when the followers play a pure NE which maximises the leader's utility, and to a {\em Pessimistic Leader-Follower Pure Nash Equilibrium} (P-LFPNE) when they seek a pure NE by which the leader's utility is minimised.

\subsubsection{The Optimistic Case}

Before focusing our attention entirely on the pessimistic case, let us briefly address the optimistic one.

An O-LFPNE can be found by solving the following {\em bilevel programming problem with $n-1$ followers}:

\everymath{\displaystyle}
\begin{equation} \label{problem:opt}
  \begin{array}{llllr}
  \max_{\substack{x_n, x_{-n}}} & \multicolumn{3}{l}{\sum_{a \in A} U_n^{a} x_1^{a_1} \, x_2^{a_2} \dots \, x_n^{a_n}}\\
    \text{s.t.}           &  x_n \in & \multicolumn{2}{l}{\Delta_n}\\
                          &  x_p \in & \argmax_{x_p} & \sum_{a \in A} U_p^{a} x_1^{a_1} \, x_2^{a_2} \dots \, x_n^{a_n} & \quad \forall p \in F\\
                          &          & \text{s.t.} & x_p \in \Delta_p \cap \{0,1\}^{m_p}.
  \end{array}
\end{equation} 
\everymath{\textstyle}

\noindent Note that, due to the integrality constraints on $x_p$ for all $p \in F$, each follower can play a single action with probability 1. By imposing the $\argmax$ constraint for each $p \in F$, the formulation guarantees that each follower plays a best-response action $a_p$, thus guaranteeing that the action profile $a_{-n} = (a_1, \dots, a_{n-1})$ with, for all $a_p \in A_p$, $a_p = 1$ if and only if $x_p^{a_p} = 1$, be an NE for the given $x_n$. It is crucial to note that the maximisation in the upper level is carried out not only w.r.t. $x_n$, but also w.r.t. $x_{-n}$. This way, if, for the chosen $x_n$, the followers' game admits multiple NEs, optimal solutions to Problem~\eqref{problem:opt} are guaranteed to contain followers' action profiles which maximise the leader's utility---thus satisfying the assumption of optimism.

As easily shown in the following proposition, computing an O-LFPNE is an easy task:
\begin{proposition}
In a normal-form game, an O-LFPNE can be computed in polynomial time by solving a multi-LP.
\end{proposition}
\begin{proof}
It suffices to enumerate, in $O(m^{n-1})$, all the followers' action profiles $a_{-n} \in A_F$ and, for each of them, solve an LP to {\em i}. check whether there is a strategy vector $x_n$ for the leader for which the action profile $a_{-n}$ is an NE and {\em ii}. find, among all such strategy vectors $x_n$, one which maximises the leader's utility. The action profile $a_{-n}$ which, with the corresponding $x_n$, yields the largest expected utility for the leader is an O-LFPNE.

Given a followers' action profile $a_{-n}$, {\em i} and {\em ii}
can be carried out in polynomial time by solving the following LP, where the second constraint guarantees that, for any of its solutions $x_n$, $a_{-n} = (a_1, \dots, a_{n-1})$ is a pure NE for the followers' game:
{\everymath{\displaystyle}
\begin{equation*}
\begin{array}{rlll}
\max_{x_n}         & \sum_{a_n \in A_n} U_n^{a_{-n},a_n} x_n^{a_n}\\
\text{s.t.} & \sum_{a_n \in A_n} U_p^{a_{-n},a_n} x_n^{a_n} \geq \sum_{a_n \in A_n} U_p^{a_1 \dots a_p' \dots a_{n-1} a_n} x_n^{a_n} & \forall p \in F, a_p' \in A_p \setminus \{a_p\}\\
             & x_n \in \Delta_n.
\end{array}
\end{equation*}
}

Note that, assuming utilities in $[1,U_{max}]$ and a binary encoding, the size of an instance of the problem is $O(m^n \lceil \log U_{max} \rceil)$ and, thus, the followers' action profiles can be enumerated in polynomial time. The claim of polynomiality of the overall algorithm follows due to linear programming problems being solvable in polynomial time. \qed
\end{proof} 



\subsubsection{The Pessimistic Case}\label{subsubsec:pess}

In the pessimistic case, the computation of a P-LFPNE amounts to solving the following {\em pessimistic bilevel problem with $n-1$ followers}:

\everymath{\displaystyle}
\begin{equation} \label{problem:pess}
  \begin{array}{llllr}
  \sup_{x_n} \min_{x_{-n}}   & \multicolumn{3}{l}{\sum_{a \in A} U_n^{a} x_1^{a_1} \, x_2^{a_2} \dots \, x_n^{a_n}}\\
    \text{s.t.}           &  x_n \in & \multicolumn{2}{l}{\Delta_n}\\
                          &  x_p \in & \argmax_{x_p} & \sum_{a \in A} U_p^{a} x_1^{a_1} \, x_2^{a_2} \dots \, x_n^{a_n} & \quad \forall p \in F\\
                          &          & \text{s.t.} & x_p \in \Delta_p \cap \{0,1\}^{m_p}.
  \end{array}
\end{equation} 
\everymath{\textstyle}

\noindent There are two differences between this problem and its optimistic counterpart: the presence of the $\min$ operator in the objective function and the fact that, rather than for a $\max$, Problem~\eqref{problem:pess} calls for a $\sup$. The former guarantees that, in the presence of more pure NEs in the followers' game for the chosen $x_n$, one which minimises the leader's utility is selected. The $\sup$ operator is introduced dbecause, as illustrated in Subsection~\ref{subsec:preliminary}, the pessimistic problem does not admit, in the general case, a maximum.

Throughout the paper, we will compactly refer to the above problem as
$$\sup_{x_n \in \Delta_n} f(x_n),$$
where $f$ is the leader's utility in the pessimistic case, defined as a function of~$x_n$. Since a pure NE may not exist for every leader's strategy $x_n$, we define $\sup_{x_n \in \Delta_n} f(x_n) = - \infty$ whenever there is no $x_n$ such that the resulting followers' game admits a pure NE. Note that $f$ is always bounded from above when assuming bounded payoffs and, thus, $\sup_{x_n \in \Delta_n} f(x_n) < \infty$.

\subsection{Some Preliminary Results}\label{subsec:preliminary}
As it is clear, since not all normal-form games admit a pure NE, a normal-form game may not admit an LFPNE. Nevertheless, assuming that the payoffs of the game are independent and follow a uniform distribution, a leader's commitment such that the resulting followers' game has at least one pure NE exists with high probability, provided that the number of players' actions is sufficiently large. This is shown in the following proposition:
\begin{proposition}
Given a normal-form game with $n$ players with independent and uniformly distributed payoffs, the probability that there exists a leader's strategy $x_n \in \Delta_n$ inducing at least one pure NE in the followers' game approaches $1$ as the number of players' actions~$m$ goes to infinity.   
\end{proposition}
\begin{proof}
In a normal-form game with independent and uniformly distributed payoffs with $n$ players, as shown in~\cite{stanford1995note}, the probability of the existence of at least one pure NE can be expressed as a function of the number of players' actions~$m$, say $\mathcal{P}(m)$, which approaches $1 - \frac{1}{e}$ for $m \rightarrow \infty$. Suppose now that we are given one such $n$-player normal-form game. Then, for every leader's action $a_n \in A_n$, let $\mathcal{P}_{a_n}(m)$ be the probability that the followers' game induced by the leader's action $a_n$ admits at least a pure NE. Since each of the followers' games resulting from the choice of $a_n$ also has independent and uniformly distributed payoffs, all the probabilities are equal, i.e., $\mathcal{P}_{a_n}(m) = \mathcal{P}(m)$ for every $a_n \in A_n$. It follows that the probability that at least one of such followers' games admits a pure NE is: 
$$
1 - \prod_{a_n \in A_n} \left( 1 - \mathcal{P}_{a_n}(m) \right)= 1 - \left( 1 - \mathcal{P}(m) \right)^m.
$$
Since, as $m$ goes to infinity, this probability approaches $1$, the probability of the existence of a leader's strategy $x_n \in \Delta_n$ which induces at least one pure NE in the followers' game also approaches $1$ for $m \rightarrow \infty$. \qed
\end{proof}

The fact that Problem~\eqref{problem:pess}
may not admit a maximum is shown by the following proposition:
\begin{proposition}\label{prop:nonexistence}
In a normal-form game, Problem~\eqref{problem:pess} may not admit a $\max$ even if the followers' game admits a pure NE for any leader's mixed strategy~$x_n$.
\end{proposition}
\begin{proof}
Consider a game with $n=3$, $A_1 = \{a_1^1,a_1^2\}$, $A_2 = \{a_2^1, a_2^2\}$, $A_3 = \{a_3^1,a_3^2\}$. The matrices reported in the following are the utility matrices for, respectively, the case where the leader plays action $a_3^1$ with probability 1, action $a_3^2$ with probability 1, or the strategy vector $x_3 = (1-\rho, \rho)$ for some $\rho \in [0,1]$ (the third matrix is the convex combination of the first two with weights $x_3$):

\hspace{-0.6cm}\begin{minipage}{0.30\textwidth}
		{\renewcommand{\arraystretch}{2}
			\begin{tabular}{r|c|c|}
				\multicolumn{1}{r}{}
				&  \multicolumn{1}{c}{$a_2^1$}
				& \multicolumn{1}{c}{$a_2^2$} \\
				\cline{2-3}
				$a_1^1$ & 1,1,0 & 2,2,5 \\
				\cline{2-3}
				$a_1^2$ & $\frac{1}{2}$,$\frac{1}{2}$,1 & 1,1,0 \\
				\cline{2-3}
				\multicolumn{1}{r}{}
				& \multicolumn{2}{c}{$ a_3^1 $}
			\end{tabular}}
\end{minipage}
\begin{minipage}{0.30\textwidth}
		{\renewcommand{\arraystretch}{2}
			\begin{tabular}{r|c|c|}
				\multicolumn{1}{r}{}
				&  \multicolumn{1}{c}{$ a_2^1 $}
				& \multicolumn{1}{c}{$ a_2^2 $} \\
				\cline{2-3}
				$ a_1^1 $ & 0,0,0 & 2,2,10 \\
				\cline{2-3}
				$ a_1^2 $ & $\frac{1}{2}$,$\frac{1}{2}$,1 & 0,0,0 \\
				\cline{2-3}
				\multicolumn{1}{r}{}
				& \multicolumn{2}{c}{$ a_3^2 $}
			\end{tabular}}
\end{minipage}
\begin{minipage}{0.30\textwidth}
		{\renewcommand{\arraystretch}{2}
			\begin{tabular}{r|c|c|}
				\multicolumn{1}{r}{}
				&  \multicolumn{1}{c}{$ a_2^1 $}
				& \multicolumn{1}{c}{$ a_2^2 $} \\
				\cline{2-3}
				$ a_1^1 $ & $1 \sminus \rho$,$1 \sminus \rho$,0 & 2,2,$5 \splus 5\rho$ \\
				\cline{2-3}
				$ a_1^2 $ & $\frac{1}{2}$,$\frac{1}{2}$,1 & $1 \sminus \rho$,$1 \sminus \rho$,0 \\
				\cline{2-3}
				\multicolumn{1}{r}{}
				& \multicolumn{2}{c}{$x_3 = (1-\rho,\rho)$}
			\end{tabular}}
\end{minipage}

In the optimistic case, as it is easy to verify, $(a_1^1,a_2^2,a_3^2)$ is the unique O-LFPNE (as it achieves the largest leader's payoff in $U_3$, a mixed strategy $x_3$ would not yield a better utility).

In the pessimistic case,
by playing $x_3 = (1-\rho,\rho)$, the leader induces the followers' game in the third matrix. For $ \rho < \frac{1}{2} $, $(a_1^1,a_2^2)$ is the unique NE, giving the leader a utility of $5+5\rho $. For $\rho \geq \frac{1}{2} $, there are two NEs, $(a_1^1,a_2^2)$ and $ (a_1^2,a_2^1)$, with a utility of, respectively, $5+5\rho$ and 1. Since, in the pessimistic case, the latter is selected, we conclude that the leader's utility is equal to $5+5 \rho$ for $\rho < \frac{1}{2}$ and to 1 for $\rho \geq \frac{1}{2}$ (see Figure~\ref{fig:nonexistence} for an illustration). Thus, Problem~\eqref{problem:pess} admits a supremum of value $5+\frac{5}{2}$, but not a maximum. \qed
\end{proof}
We remark that the result in Proposition~\ref{prop:nonexistence} is in line with a similar result
shown in~\cite{leaderfollower} for the single-follower case, as well as those which hold for more general pessimistic bilevel problems~\cite{zemkoho2016}.

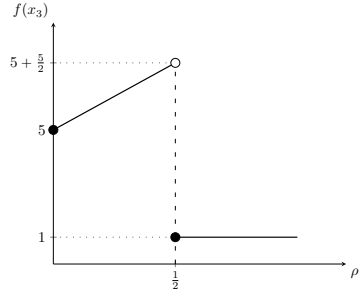
\begin{figure}
\begin{center}
		 \resizebox{.4\textwidth}{!}{%
                         \begin{tikzpicture}
	\begin{axis}[ 
	    axis x line=center,
	    axis y line=center,
	    xtick={0.5},
	    xticklabels={$ \frac{1}{2} $},
	    ytick={1,5,7.5},
	    yticklabels={1,5,$ 5 + \frac{5}{2} $},
	    xlabel={$ \rho $},
	    ylabel={$ f(x_3)$},
	    every axis/.append style={font=\normalsize},
	    xlabel style={font=\normalsize, below right},
	    ylabel style={font=\normalsize, above left},
	    xmin=0,
	    xmax=1.2,
	    ymin=0,
	    ymax=9]
	  \addplot[color=black,thick] coordinates { (0,5) (0.5,7.5) };
	  \addplot[color=black,thick] coordinates { (0.5,1) (1,1) };
	  \addplot[only marks,mark=*,mark size=3pt] coordinates { (0,5) (0.5,1) };
	  \addplot[only marks,mark=*,mark size=3pt,mark options={fill=white}] coordinates { (0.5,7.5) };
	  \addplot[color=black,style=loosely dashed] coordinates { (0.5,0) (0.5,7.5) };
	  \addplot[color=black,style=loosely dotted] coordinates { (0,7.5) (0.5,7.5) };
	  \addplot[color=black,style=loosely dotted] coordinates { (0,1) (0.5,1) };
	\end{axis}
      \end{tikzpicture}
}%
\end{center}
\caption{The leader's utility in the normal-form game in the proof of Proposition~\ref{prop:nonexistence}, showing that Problem~\eqref{problem:pess} may not admit a maximum.}
\label{fig:nonexistence}
\end{figure}

The relevance of computing a pessimistic LFPNE is highlighted by the following proposition:
\begin{proposition}\label{prop:arbitrarily}
In normal-form games, the leader's utility in a P-LFPNE can be arbitrarily worse than that in an O-LFPNE. Moreover, the utility that is obtained after perturbing the leader's strategy in an O-LFPNE can be arbitrarily worse than that one in a P-LFPNE.
\end{proposition}
\begin{proof}
Consider the following normal-form game, with $n=3$, $A_1=\{a_1^1,a_1^2\}$, $A_2=\{a_2^1,a_2^2\}$, $A_3=\{a_3^1,a_3^2\}$, parameterised by $\mu > 1$:
    
\hspace{-0.6cm}\begin{minipage}{0.30\textwidth}
  {\renewcommand{\arraystretch}{2}
    \begin{tabular}{r|c|c|}
      \multicolumn{1}{r}{}
      &  \multicolumn{1}{c}{$a_2^1$}
      & \multicolumn{1}{c}{$a_2^2$} \\
      \cline{2-3}
      $a_1^1$ & 1,1,0 & $\frac{1}{2}$,$\frac{1}{2}$,0 \\
      \cline{2-3}
      $a_1^2$ & 2,2,1 & 0,0,0 \\
      \cline{2-3}
      \multicolumn{1}{r}{}
      & \multicolumn{2}{c}{$ a_3^1 $}
  \end{tabular}}
\end{minipage}
\begin{minipage}{0.30\textwidth}
  {\renewcommand{\arraystretch}{2}
      \begin{tabular}{r|c|c|}
	\multicolumn{1}{r}{}
	&  \multicolumn{1}{c}{$ a_2^1 $}
	& \multicolumn{1}{c}{$ a_2^2 $} \\
	\cline{2-3}
	$ a_1^1 $ & 0,0,0 & $\frac{1}{2}$,$\frac{1}{2}$,$4 \mu$ \\
	\cline{2-3}
	$ a_1^2 $ & 2,2,$\mu$ & 1,1,0 \\
	\cline{2-3}
	\multicolumn{1}{r}{}
	& \multicolumn{2}{c}{$ a_3^2 $}
			\end{tabular}}
\end{minipage}
\begin{minipage}{0.30\textwidth}
  {\renewcommand{\arraystretch}{2}
    \begin{tabular}{r|c|c|}
      \multicolumn{1}{r}{}
      &  \multicolumn{1}{c}{$ a_2^1 $}
      & \multicolumn{1}{c}{$ a_2^2 $} \\
      \cline{2-3}
      $ a_1^1 $ & $1 \sminus \rho$,$1 \sminus \rho$,0 & $\frac{1}{2}$,$\frac{1}{2}$,$4 \mu \rho$ \\
      \cline{2-3}
      $ a_1^2 $ & 2,2,$1 \splus \rho(\mu \sminus 1)$  & $\rho$,$\rho$,0 \\
      \cline{2-3}
      \multicolumn{1}{r}{}
      & \multicolumn{2}{c}{$x_3 = (1 \sminus \rho,\rho)$}
  \end{tabular}}
\end{minipage}

Let $x_3 = (1-\rho,\rho)$. The followers' game admits the NE $(a_1^2, a_2^1)$ for all values of $\rho$ (with leader's utility $1+\rho(\mu-1)$), as well as a second one, $(a_1^1, a_2^2)$, for $\rho = \frac{1}{2}$ (with leader's utility $2 \mu$). Therefore, the game admits a unique O-LFPNE, achieved at $\rho = \frac{1}{2}$ (utility $2\mu$), and a unique P-LFPNE, achieved at $\rho=1$ (utility $\mu$). See Figure~\ref{fig:perturbation_plot} for an illustration of the leader's utility function.

To show the first part of the claim, it suffices to observe that, by letting $\mu \rightarrow \infty$, the difference in utility between O-LFPNE and P-LFPNE, equal to $\mu$, becomes arbitrarily large.

As to the second part of the claim, note that after perturbing the value that $x_3$ takes in the unique O-LFPNE by any $\epsilon \in [-\frac{1}{2}, \frac{1}{2}]$ with $\epsilon \neq 0$ we obtain a leader's utility of $(1+\mu)/2 + (\mu-1)\epsilon$, whose difference w.r.t. the utility of~$\mu$ in the unique P-LFPNE is again arbitrarily large for $\mu \rightarrow \infty$.
\qed
\end{proof}

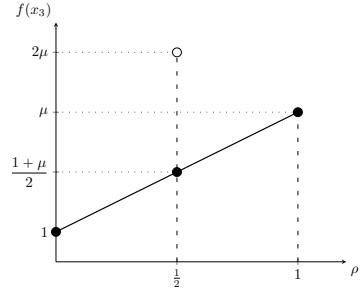
\begin{figure}[h!]
	\centering
		\resizebox{0.4\textwidth}{!}{%
			\begin{tikzpicture}
			\begin{axis}[ 
			axis x line=center,
			axis y line=center,
			xtick={0.5, 1},
			xticklabels={$ \frac{1}{2} $, $1$},
			ytick={1,3,5,7},
			yticklabels={1,$\dfrac{1+\mu}{2}$,$\mu$,$2 \mu$},
			xlabel={$ \rho $},
			ylabel={$f(x_3)$},
			every axis/.append style={font=\normalsize},
			xlabel style={font=\normalsize, below right},
			ylabel style={font=\normalsize, above left},
			xmin=0,
			xmax=1.2,
			ymin=0,
			ymax=8]
			\addplot[color=black,thick] coordinates { (0,1) (1,5) };
			\addplot[only marks,mark=*,mark size=3pt] coordinates { (0,1) (0.5,3) (1,5)};
			\addplot[only marks,mark=*,mark size=3pt,mark options={fill=white}] coordinates { (0.5,7) };
			\addplot[color=black,style=loosely dashed] coordinates { (0.5,0) (0.5,7) };
			\addplot[color=black,style=loosely dashed] coordinates { (1,0) (1,5) };
			\addplot[color=black,style=loosely dotted] coordinates { (0,7) (0.5,7) };
			\addplot[color=black,style=loosely dotted] coordinates { (0,3) (.5,3) };
			\addplot[color=black,style=loosely dotted] coordinates { (0,5) (1,5) };
			\end{axis}
			\end{tikzpicture}
		}%
	        \caption{The leader's utility in the normal-form game in the proof of Proposition~\ref{prop:arbitrarily}, plotted as a function of $\rho$, where the leader's strategy is $x_3=(1 - \rho,\rho)$.} 
	\label{fig:perturbation_plot}
\end{figure}

\section{Computational Complexity}\label{sec:reduction}

Let P-LFPNE-s be the search version of the problem of computing a P-LFPNE. In this section, we study the computational complexity of solving P-LFPNE-s for normal-form games. In particular, we show, in Subsection~\ref{sub_sec:hard}, that P-LFPNE-s is \textsf{NP}-hard for $n \geq 3$ (i.e., with at least two followers), and, in Subsection~\ref{sub_sec:apx}, that P-LFPNE-s is not in Poly-\textsf{APX} for $n \geq 4$ (i.e., for games with at least three followers), unless \textsf{P} = \textsf{NP}. We introduce two reductions, a non approximation-preserving one which is valid for $n \geq 3$ and another one, only valid for $n \geq 4$, but approximation-preserving.

In decision form, the problem of computing a P-LFPNE reads:

\begin{definition}[P-LFPNE-d]\label{def:ourprob}
Given a normal-form game with $n \geq 3$ players and a finite number $K$, is there a P-LFPNE where the leader achieves a utility greater than or equal to $K$?
\end{definition}

We show, in Section~\ref{sub_sec:hard}, that P-LFPNE-d is \textsf{NP}-complete via a polynomial-time reduction of
Independent Set (IND-SET), one of Karp's original 21 \textsf{NP}-complete problems~\cite{karp1972reducibility}, to it. In decision form, IND-SET reads:
\begin{definition}[IND-SET-d]\label{def:indset}
Given an undirected graph $G=(V,E)$ and an integer $J \leq |V|$, does $G$ contain an \emph{independent set} (a subset of vertices $V' \subseteq V: \forall u,v \in V'$, $\{u,v\} \notin E$) of size greater than or equal to $J$?
\end{definition}

We prove, in Subsection~\ref{sub_sec:apx}, the inapproximability of P-LFPNE-s for the case with at least three followers via a polynomial-time reduction of
3-SAT, another of Karp's 21 \textsf{NP}-complete problems~\cite{karp1972reducibility}, to P-LFPNE-d. 3-SAT reads:
\begin{definition}[3-SAT]\label{def:3sat}
  Given a collection $C=\{\phi_1,\ldots,\phi_t\}$ of clauses (disjunctions of literals) on a finite set $V$ of boolean variables with $|\phi_c|=3$ for $1 \leq c \leq t$, is there a truth assignment for $V$ which satisfies all the clauses in $C$?
\end{definition}

\subsection{\textsf{NP}-Completeness}\label{sub_sec:hard}

Before presenting our reduction, we introduce the following class of normal-form games:

\begin{definition}\label{def:gamma}
Given two
rational numbers
$b$ and $c$, with $1 > c > b > 0$, and an integer $r \geq 1$, let $\Gamma_b^c(r)$ be a class of normal-form games with three players ($n=3$), the first two having $r+1$ actions each, with action sets $A_1 = A_2 = A = \{1,...,r,\chi\}$, the third one having $r$ actions, with action set $A_3 = A \setminus \{\chi\}$, and such that, for every third player's action $a_3 \in A \setminus \{\chi\} $, the other players play a game where:
\begin{itemize}
\item the payoffs on the main diagonal (where both players play the same action) satisfy $U_1^{a_3 a_3 a_3} \sequal U_2^{a_3 a_3 a_3} \sequal 1, U_1^{\chi \chi a_3} \sequal c, U_2^{\chi \chi a_3} \sequal b $ and, for any $a_1 \in A \setminus \{a_3,\chi\}$, $U_1^{a_1 a_1 a_3} \sequal U_2^{a_1 a_1 a_3} \sequal 0$;
\item for every $a_1,a_2 \in A \setminus \{\chi\}$ with $a_1 \neq a_2$, $U_1^{a_1 a_2 a_3} \sequal U_2^{a_1 a_2 a_3} = b$;
\item for every $a_2 \in A \setminus \{\chi\}$, $U_1^{\chi a_2 a_3} \sequal c $ and $ U_2^{\chi a_2 a_3} \sequal 0$;
\item for every $a_1 \in A \setminus \{\chi\}$, $U_1^{a_1 \chi a_3} \sequal 1 $ and $ U_2^{a_1 \chi a_3} \sequal 0$.
\end{itemize}
No restrictions are imposed on the third player's payoffs.
\end{definition}

\begin{figure}[!htp]
	\centering
	{\renewcommand{\arraystretch}{2}
		\begin{tabular}{r|c|c|c|c|}
			\multicolumn{1}{r}{}
			&  \multicolumn{1}{c}{$1$}
			& \multicolumn{1}{c}{$2$} 
			& \multicolumn{1}{c}{$3$} 
			& \multicolumn{1}{c}{$\chi$} \\
			\cline{2-5}
			$1$ & $1,1,0$ & $b,b,0$ & $b,b,0$ & $1,0,0$ \\
			\cline{2-5}
			$2$ & $b,b,0$ & $0,0,1$ & $b,b,0$ & $1,0,0$ \\
			\cline{2-5}
			$3$ & $b,b,0$ & $b,b,0$ & $0,0,1$ & $1,0,0$ \\
			\cline{2-5}
			$\chi$ & $c,0,0$ & $c,0,0$ & $c,0,0$ & $c,b,0$ \\
			\cline{2-5}
			\multicolumn{1}{r}{}
			& \multicolumn{4}{c}{$ 1 $}
		\end{tabular}}
		{\renewcommand{\arraystretch}{2}
			\begin{tabular}{r|c|c|c|c|}
				\multicolumn{1}{r}{}
				&  \multicolumn{1}{c}{$1$}
				& \multicolumn{1}{c}{$2$} 
				& \multicolumn{1}{c}{$3$} 
				& \multicolumn{1}{c}{$\chi$} \\
				\cline{2-5}
				$1$ & $0,0,1$ & $b,b,0$ & $b,b,0$ & $1,0,0$ \\
				\cline{2-5}
				$2$ & $b,b,0$ & $1,1,0$ & $b,b,0$ & $1,0,0$ \\
				\cline{2-5}
				$3$ & $b,b,0$ & $b,b,0$ & $0,0,-\frac{1}{c}-1$ & $1,0,0$ \\
				\cline{2-5}
				$\chi$ & $c,0,0$ & $c,0,0$ & $c,0,0$ & $c,b,0$ \\
				\cline{2-5}
				\multicolumn{1}{r}{}
				& \multicolumn{4}{c}{$ 2 $}
			\end{tabular}}
			{\renewcommand{\arraystretch}{2}
				\begin{tabular}{r|c|c|c|c|}
					\multicolumn{1}{r}{}
					&  \multicolumn{1}{c}{$1$}
					& \multicolumn{1}{c}{$2$} 
					& \multicolumn{1}{c}{$3$} 
					& \multicolumn{1}{c}{$\chi$} \\
					\cline{2-5}
					$1$ & $0,0,1$ & $b,b,0$ & $b,b,0$ & $1,0,0$ \\
					\cline{2-5}
					$2$ & $b,b,0$ & $0,0,-\frac{1}{c}-1$ & $b,b,0$ & $1,0,0$ \\
					\cline{2-5}
					$3$ & $b,b,0$ & $b,b,0$ & $1,1,0$ & $1,0,0$ \\
					\cline{2-5}
					$\chi$ & $c,0,0$ & $c,0,0$ & $c,0,0$ & $c,b,0$ \\
					\cline{2-5}
					\multicolumn{1}{r}{}
					& \multicolumn{4}{c}{$ 3 $}
				\end{tabular}}
		\caption{A $\Gamma_b^c(r)$ game with $r=3$. The third player (the leader) selects a matrix, while the first and the second players (the followers) select rows and columns, respectively. The third player's payoffs are defined starting from the graph in Figure~\ref{fig:graph_indset}, as explained in the proof of Theorem~\ref{thm:hard}.}
		\label{fig:game_indset}
\end{figure}


The special feature of $\Gamma_b^c(r)$ games, see Figure~\ref{fig:game_indset} for an illustration of one such game with $r=3$, parametric in $b$ and $c$, is that, no matter which mixed strategy the third player (the leader) commits to, only diagonal outcomes, with the exception of $(\chi,\chi)$, can be pure NEs in the resulting followers' game. Moreover, for every subset of diagonal outcomes, there is a leader's strategy such that this subset precisely corresponds to the set of all pure NEs in the followers' game, as formally stated by the following proposition:
\begin{proposition}\label{prop:equlibria}
A $\Gamma_b^c(r)$ game with $c \leq \frac{1}{m} $ for all $S \subseteq \{(a_1,a_1) : a_1 \in A \setminus \{\chi\}\}$ with $S\neq \emptyset$ admits a leader's strategy $x_3 \in \Delta_3$ such that the outcomes $(a_1,a_1) \in S$ are the only pure NEs in the resulting followers' game.
\end{proposition}

\begin{proof}
First, observe that the followers' payoffs that are not on the main diagonal are independent of the leader's strategy $x_3$. Thus, outcomes $ (a_1,a_2) $, for any $a_1,a_2 \in A \setminus \{\chi\}$ with $a_1 \neq a_2$, cannot be NEs, as the first follower would deviate by playing action $\chi$ so to obtain a utility $c > b$.
Analogously, any outcome $(\chi,a_2)$, with $a_2 \in A \setminus \{\chi\}$, cannot be NE because the second follower would deviate by playing $\chi$ (since $b > 0$). The same holds for outcomes $(a_1,\chi)$ with $a_1 \in A \setminus \{\chi\}$, since the second follower would be better off playing another action (as $b > 0$).
The last outcome on the diagonal, $(\chi,\chi)$, cannot be an NE either, as the first follower would deviate from it (as she would get $c$ in it, while she can obtain $1 > c$ by deviating).

As a result, the only outcomes which can be pure NEs are those in $\{(a_1,a_1) : a_1 \in A \setminus \{\chi\} \}$. Clearly, when the leader plays a pure strategy $a_3$, the unique pure NE in the followers' game is $(a_3,a_3)$ as, due to providing the followers with their maximum payoff, they would not deviate from it. Outcomes $(a_1,a_1)$ with $a_1 \in A \setminus \{\chi,a_3\}$ are not NEs as, with them, the first follower would get $0 < c$. In general, if the leader plays an arbitrary mixed strategy $x_3 \in \Delta_3$, the resulting followers' game is such that the payoffs in $(a_3,a_3)$, with $a_3 \in A \setminus \{\chi\} $, are $(x_3^{a_3},x_3^{a_3})$. Noticing that $(a_3,a_3)$ is an equilibrium if and only if $x_3^{a_3} \geq c$ (as, otherwise, the first follower would deviate by playing action $\chi$), we conclude that the set of pure NEs in the followers' game is $S = \{(a_3,a_3) : x_3^{a_3} \geq c\}$.
 
In order to guarantee that, for every possible $S \subseteq \{(a_1,a_1) : a_1 \in A\setminus \{\chi\}\}$ with $S\neq \emptyset$, there is a leader's strategy such that $S$ contains all the pure NEs of the followers' game, we must allow the diagonal outcomes to be all (simultaneously) equilibria by properly choosing the value of $c$. This is done by imposing that, when the leader plays $ x_3 = (\frac{1}{r},\frac{1}{r},...,\frac{1}{r}) $, all outcomes in $\{(a_1,a_1) : a_1 \in A \setminus \{\chi\} \}$ are NEs, which is obtained by selecting $c \leq \frac{1}{r}$. \qed
\end{proof}

Notice that, in a $\Gamma_b^c(r)$ game with $c \leq \frac{1}{r}$, the followers' game always admits a pure NE for any leader's commitment $x_3$. 
Graphically, as shown in Figure~\ref{fig:simplex} for $r=3$, the leader's strategy space, $\Delta_3$, is partitioned into $2^{r}-1$ regions, each corresponding to a subset of $\{(a_1,a_1) : a_1 \in A \setminus \{\chi\} \}$ containing those diagonal outcomes which are the only NEs in the followers' game. 
Hence, in a $\Gamma_b^c(r)$ game with $c \leq \frac{1}{r}$, the number of combinations of outcomes which may constitute the set of NEs in the followers' game is exponential in $r$, and, thus, in the size of the game instance.

\begin{figure}[!htp]
	\centering
\scalebox{.6}{\begin{tikzpicture}[line cap=round,line join=round,>=triangle 45,x=1.0cm,y=1.0cm]
\clip(-1,-0.3032019954410076) rectangle (9,7.594624123104314);
\draw [line width=1.6pt] (0.,0.)-- (4.,6.928203230275509);
\draw [line width=1.6pt] (8.,0.)-- (4.,6.928203230275509);
\draw [line width=1.6pt] (8.,0.)-- (0.,0.);
\draw [line width=1.6pt,dash pattern=on 4pt off 4pt] (0.7146758902882556,1.2378549529237797)-- (7.285324109711745,1.2378549529237797);
\draw [line width=1.6pt,dash pattern=on 4pt off 4pt] (1.4302098610758234,0.)-- (4.715104930537911,5.689605157840833);
\draw [line width=1.6pt,dash pattern=on 4pt off 4pt] (6.571877155387122,0.)-- (3.2859385776935603,5.69141256711586);
\draw (0.8853971371183099,0.8309932629798307) node[anchor=north west] {\parbox{2.373792274626238 cm}{{\small\bf \textsf{A}}}};
	\draw(1.09865317635063,0.6298998642592937) circle (0.3cm);
	\draw(2.2107707320004404,2.625169596454556) circle (0.3cm);
	\draw(5.77608877805425,2.641524266390582) circle (0.3cm);
	\draw(3.3719522974583334,0.6135451943232683) circle (0.3cm);
	\draw(4.647616552468412,0.6135451943232683) circle (0.3cm);
	\draw(2.783184179761374,3.622804462552186) circle (0.3cm);
	\draw(5.252739340101397,3.606449792616159) circle (0.3cm);
	\draw(4.009784424963373,5.6998475444275805) circle (0.3cm);
	\draw(6.937270343512141,0.6135451943232683) circle (0.3cm);
	\draw(3.977075085091319,2.90319898536701) circle (0.3cm);
	\draw(3.5027896569465464,2.167238838245807) circle (0.3cm);
	\draw(4.435005843300065,2.167238838245807) circle (0.3cm);
	\draw (1.9975146927681217,2.846262995175092) node[anchor=north west] {\parbox{2.373792274626238 cm}{{\small\bf \textsf{A}}}};
		\draw (3.1586962582260133,0.8346385930438039) node[anchor=north west] {\parbox{2.373792274626238 cm}{{\small\bf \textsf{A}}}};
			\draw (3.289533617714227,2.3883322369663434) node[anchor=north west] {\parbox{2.373792274626238 cm}{{\small\bf \textsf{A}}}};
				\draw (2.5635734705930277,3.8338978612727225) node[anchor=north west] {\parbox{2.373792274626238 cm}{{\small\bf \textsf{B}}}};
					\draw (3.7574643759229726,3.114292384087546) node[anchor=north west] {\parbox{2.373792274626238 cm}{{\small\bf \textsf{B}}}};
						\draw (5.033128630933051,3.8175431913366956) node[anchor=north west] {\parbox{2.373792274626238 cm}{{\small\bf \textsf{B}}}};
							\draw (3.790173715795026,5.89458627321209) node[anchor=north west] {\parbox{2.373792274626238 cm}{{\small\bf \textsf{B}}}};
								\draw (6.707659634343795,0.8182839231077772) node[anchor=north west] {\parbox{2.373792274626238 cm}{{\small\bf \textsf{C}}}};
									\draw (4.418005843300065,0.8182839231077772) node[anchor=north west] {\parbox{2.373792274626238 cm}{{\small\bf \textsf{C}}}};
										\draw (5.546478068885903,2.846262995175092) node[anchor=north west] {\parbox{2.373792274626238 cm}{{\small\bf \textsf{C}}}};
											\draw (4.205395134131719,2.3883322369663434) node[anchor=north west] {\parbox{2.373792274626238 cm}{{\small\bf \textsf{C}}}};
\end{tikzpicture}}
\caption{A $\Gamma_b^c(r)$ game with $r=3$ and $c \leq \frac{1}{r}$. The leader's strategy space $\Delta_3$ is partitioned into $2^{r}-1$ regions, one per subset of $\{(a_1,a_1) : a_1 \in A \setminus \{\chi\} \}$ (the three NEs in the followers' game, $(1,1), (2,2)$, and $(3,3)$, are labelled \textsf{A}, \textsf{B}, \textsf{C}).}
\label{fig:simplex}
\end{figure}
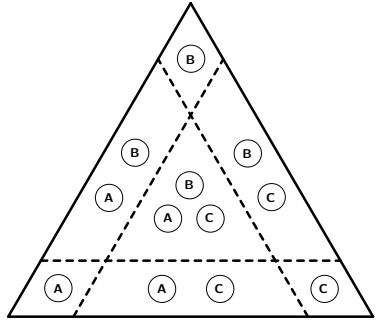

Relying on Proposition~\ref{prop:equlibria}, we can establish the following result:
\begin{theorem}\label{thm:hard}
\emph{P-LFPNE-d} is strongly \textsf{NP}-complete even for $n=3$.
\end{theorem}

\begin{proof}
For the sake of clarity, we split the proof in some steps.

{\bf Mapping.} Given an instance of IND-SET, i.e., an undirected graph $G=(V,E)$ and a positive integer $J$, we construct $\Gamma(G)$, a special instance of P-LFPNE-d of class $\Gamma_b^c(r)$, as follows. Assuming an arbitrary labeling of the vertices $\{v_1,v_2,...,v_r\}$, let $\Gamma(G)$ be an instance of $\Gamma_b^c(r)$ with $c < \frac{1}{(r+1)^2} < \frac{1}{r}$ and $0 <b < c < 1$, where each action $a_1 \in A \setminus\{\chi\}$ is associated with a vertex $v_{a_1} \in V$. In compliance with Definition~\ref{def:gamma}, in which no constraints are specified for the leader payoffs, we define:
\begin{itemize}
\item for any pair of vertices $v_{a_1},v_{a_2} \in V$: $U_3^{a_1 a_1 a_2} = U_3^{a_2 a_2 a_1} = - \frac{1}{c} - 1$ if $\{v_{a_1},v_{a_2}\} \in E$, and $U_3^{a_1 a_1 a_2} = U_3^{a_2 a_2 a_1} = 1$ otherwise;
\item for every $a_3 \in A \setminus \{\chi\}$: $U_3^{a_3 a_3 a_3} = 0$ and $U_3^{\chi \chi a_3} = 0$;
\item for every $a_3 \in A \setminus \{\chi\} $ and for every $a_1, a_2 \in A $ with $a_1 \neq a_2$: $U_3^{a_1 a_2 a_3} = U_3^{a_2 a_1 a_3} = 0$.
\end{itemize}
As an example, Figure~\ref{fig:graph_indset}
illustrates an instance of IND-SET from which the game depicted in Figure~\ref{fig:game_indset} is obtained by reduction.
Finally, let $ K = \frac{J - 1}{J} $. Note that, as it is clear, this transformation can be carried out in time polynomial in the number of vertices $|V|=r$.

\begin{figure}[!htp]
	\centering
	\begin{tikzpicture}[auto, node distance=2cm, every loop/.style={},
			thick,main node/.style={circle,draw,font=\Large}]
			\node[main node] (1) {$v_1$};
			\node[main node] (2) [below left of=1] {$v_2$};
			\node[main node] (3) [below right of=1] {$v_3$};
			\path
			(3) edge (2);
	\end{tikzpicture}
	\caption{An undirected graph $G=(V,E)$, where $V=\{ v_1,v_2,v_3 \}$ and $E=\{ \{v_2,v_3 \} \}$.}
	\label{fig:graph_indset}
\end{figure}
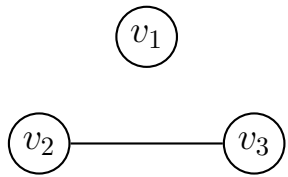

{\bf If.} We show that, if the graph $ G $ contains an independent set of size greater than or equal to $J$, then $ \Gamma(G)$ admits a P-LFPNE with leader's utility greater than or equal to $K$. Let $V^*$ be an independent set with $|V^*| = J$. Consider the case in which outcomes $(a_1,a_1)$, with $v_{a_1} \in V^*$, are the only pure NEs in the followers' game, and suppose that the leader's strategy $x_3$ is $x_3^{a_3} = \frac{1}{|V^*|}$ if $v_{a_3} \in V^*$ and $x_3^{a_3} = 0$ otherwise. Since, by construction, $U_3^{a_1 a_1 a_3} = 1$ for all $a_3 \in A \setminus \{\chi,a_1\}$, the leader's utility at an equilibrium $(a_1,a_1)$ is:
$$
\sum_{a_3 \in A \setminus \{\chi\}} U_3^{a_1 a_1 a_3} x_3^{a_3} = \sum_{a_3 \in A \setminus \{\chi,a_1\}} U_3^{a_1 a_1 a_3} x_3^{a_3} = \sum_{a_3 \in A \setminus \{\chi,a_1\}} x_3^{a_3} = \frac{|V^*|-1}{|V^*|} = K.
$$

{\bf Only if.} We show that, if $\Gamma(G)$ admits a P-LFPNE with leader's utility greater than or equal to $K$, then $G$ contains an independent set of size greater than or equal to $J$. Due to Proposition~\ref{prop:equlibria}, at any P-LFPNE the leader plays a strategy $\bar x_3$ inducing a set of pure NEs in the followers' game corresponding to $S^* = \{(a_3,a_3) : \bar x_3^{a_3} \geq c\}$.
We now show that, in a P-LFPNE, the leader would never play two actions $a_1,a_2 \in A \setminus \{\chi\}$, with $\{v_{a_1},v_{a_2}\} \in E$, with probability greater than or equal to $c$. By contradiction, suppose the leader's equilibrium strategy $\bar x_3$ is such that $\bar x_3^{a_1}, \bar x_3^{a_2} \geq c$. When the followers play the equilibrium $(a_1,a_1)$ (the same holds for $(a_2,a_2)$), the leader's utility is:
$$
\sum_{a_3 \in A \setminus \{\chi\}} U_3^{a_1 a_1 a_3} \bar x_3^{a_3} = \sum_{a_3 \in A \setminus \{\chi,a_1,a_2\}} U_3^{a_1 a_1 a_3} \bar x_3^{a_3} + \bar x_3^{a_2} (- \frac{1}{c} - 1).
$$
In the right-hand side, the first term is $<1$ (as the leader's payoffs are $\leq 1$ and $\sum_{a_3 \in A \setminus \{\chi,a_1,a_2\}} \bar x_3^{a_3} = 1 - \bar x_3^{a_1} - \bar x_3^{a_2} < 1 $, since $\bar x_3^{a_1} , \bar x_3^{a_2} \geq c$). The second term is less than or equal to $c (- \frac{1}{c} - 1) = -1 - c$ (as $\bar x_3^{a_2} \geq c$), which is strictly less than $-1$. It follows that, since $(a_1,a_1)$ (or, equivalently, $(a_2,a_2)$) always provides the leader with a negative utility, she would never play $\bar x_3$ in an equilibrium. This is because, by playing a pure strategy, she would obtain a utility of at least zero (as, when she plays a pure strategy, the followers' game admits a unique pure NE giving her a zero payoff).
As a result, for any action $a_3$ such that $\bar x_3^{a_3} \geq c$, we have $U_3^{a_3 a_3 a_3} = 0$, and $U_3^{a_1 a_1 a_3} = 1$ for every $a_1$ such that $\bar x_3^{a_1} \geq c$ (since $v_{a_1}$ and $v_{a_3}$ are not connected by an edge).

Now, let us make the following assumption.

{\em Assumption:} the leader either plays an action with probability greater than or equal to $c$ or she does not play it at all.

If this is the case, then the leader's utility at an equilibrium $(a_3,a_3) \in S^*$ is $1 - \bar x_3^{a_3}$.
Since, due to the pessimistic assumption, the leader maximises her utility in the worst NE, her best choice is to select an $\bar x_3$ such that all NEs yield the same utility, that is: $\bar x_3^{a_1} = \bar x_3^{a_2}$ for every $a_1,a_2$ with $(a_1,a_1), (a_2,a_2) \in S^*$.
This results in the leader playing all actions $a_3$ such that $(a_3,a_3) \in S^*$ with the same probability $\bar x_3^{a_3} = \frac{1}{|S^*|}$, obtaining a utility of $\frac{|S^*|-1}{|S^*|} = K$.
Therefore, the vertices in the set $\{v_{a_3} : (a_3,a_3) \in S^* \}$ form an independent set of $G$ of size $|S^*|=J$.

We now show that,
if $c < \frac{4}{(r+1)^2}$,
the previous assumption always holds, i.e., in any P-LFPNE, the leader is not better off playing any action with probability less than $c$.
Observe that, without imposing any constraint on $c$, except for $c \leq \frac{1}{r}$, the aforementioned assumption may not hold true in presence of isolated vertices. Indeed, suppose that $v_{a_1}$ is the only isolated vertex in $G$ and $r \geq 3$. Then, as we show next, there is a P-LFPNE in which the leader plays a strategy $x_3$ such that, for every $a_3 \in A \setminus \{a_1,\chi\}$, $x_3^{a_3}=c'$, where $c' = c -\epsilon$ for some $\epsilon > 0$, while $x_3^{a_1}=1 - c' (r-1)$. Since the latter probability is greater than $c$ by definition (as it is always greater than or equal to $\frac{1}{r}$), the unique NE for the followers is $(a_1,a_1)$, providing the leader with a utility of $1 - x_3^{a_1} = c' (r-1)$, which approaches $\frac{r-1}{r}$ for $c' \rightarrow \frac{1}{r}$. 
Assuming $c=\frac{1}{r}$, for $c' \rightarrow \frac{1}{r}$ the strategy $x_3$ is part of a P-LFPNE, since, as previously shown, the leader cannot get more than $\frac{r-2}{r-1}$ without playing actions with probability smaller than $c$.
For instance, consider the game in Figure~\ref{fig:game_indset} that is obtained from the graph $G$ in Figure~\ref{fig:graph_indset}. As one can see, $v_1$ is the only isolated vertex in $G$, and, as a consequence, the strategy $x_3$ such that $x_3^1 = 1 - 2 c'$ and $x_3^2 = x_3^3 = c'$ is part of a P-LFPNE, for $c' \rightarrow \frac{1}{3}$.
 
In general, let us denote by $\ell$ the number of isolated vertices in $G$, and assume that the other $r - \ell \geq 2$ vertices form a complete graph. This corresponds to the worst case
as, for it, the leader cannot get a utility larger than $\frac{\ell}{\ell+1}$ without playing some actions with probability smaller than $c$, but, at the same time, she could get more by uniformly playing the $\ell$ actions associated with the isolated vertices, each with probability $\alpha = \frac{1 - (r - \ell) c'}{\ell}$, while playing with probability $c'$ the other $r - \ell$ actions. If this is the case, the leader's utility is:
$$
(\ell - 1) \alpha + (r - \ell) \ c' = \frac{(\ell - 1) + (r - \ell) c'}{\ell}.
$$
Thus, in order for the assumption to hold true, we require $\frac{(\ell - 1) + (r - \ell) c'}{\ell} < \frac{\ell}{\ell+1}$ for every $\ell \in \{1,...,r-2\}$, which implies that $c$ must satisfy:
$$
(\ell + 1) (\ell -1) + (\ell + 1) (r - \ell) c < \ell^2,
$$
in which we upper bounded $c'$ by $c$. The above condition is satisfied, in turn, if and only if:
$$
c < \frac{1}{(\ell +1) (r - \ell)}.
$$
We deduce that $c$ satisfies the condition for all $\ell$ whenever $c \leq \frac{4}{(r+1)^2}$. The latter is the minimum value taken by $\frac{1}{(\ell +1) (r - \ell)}$, achieved at $\ell = \frac{r-1}{2}$, where $\frac{(\ell +1)-(r-\ell)}{(\ell+1)^2(r-\ell)^2}$, the derivative of $\frac{1}{(\ell +1) (r - \ell)}$, which is a strictly convex function of $\ell$, vanishes. Given that, according to our definition, $c < \frac{1}{(r+1)^2} < \frac{4}{(r+1)^2}$, we obtain that the condition is always satisfied, implying that the leader either plays an action with probability at least~$c$ or she never plays such action. The reduction, thus, is complete.

{\bf \textsf{NP} membership.} 
Since, given a triple $(a_1,a_2,x_3)$, we can verify in polynomial time whether $(a_1,a_2)$ is an NE in the followers' game induced by $x_3$ and whether, when playing $(a_1,a_2,x_3)$, the leader's utility is at least $K$, we deduce that P-LFPNE-d belongs to \textsf{NP}. Thus, the problem is strongly \textsf{NP}-complete due to IND-SET being strongly \textsf{NP}-complete. \qed
\end{proof}

\subsection{Inapproximability}\label{sub_sec:apx}

We
show now that the problem of computing a P-LFNE is not only \textsf{NP}-hard, but it is also difficult to
approximate even in the case of only three followers.
Since the reduction from IND-SET
which we gave in Theorem~\ref{thm:hard} is not approximation-preserving, we propose a new one based on 3-SAT (see Definition~\ref{def:3sat}).

In the following, given a literal $l$
(an occurrence of a variable, possibly negated), we define $v(l)$ as its corresponding variable. Moreover, for a generic clause $$\phi=l_1 \vee l_2 \vee l_3,$$ we denote the ordered set of possible truth assignments to the variables, namely, $x=v(l_1),y=v(l_2)$, and $z=v(l_3)$, by $$L_{\phi}=\{x y z, x y \bar z, x \bar y z, x \bar y \bar z, \bar x y z, \bar x y \bar z, \bar x \bar y z, \bar x \bar y \bar z\},$$ where, in each truth assignment, a variable is set to 1 if positive and to 0 if negative.
Given a generic 3-SAT instance, we build a corresponding normal-form game as detailed in the following definition.

\begin{definition}\label{def:sat_gadget}
  Given a 3-SAT instance where $C=\{\phi_1,\ldots,\phi_t\}$ is a collection of clauses and $V=\{v_1,\ldots,v_r\}$ is a set of Boolean variables, let $\Gamma(C,V)$ be a normal-form game with four players ($n=4$) defined as follows. The fourth player has an action for each variable in $V$ plus an additional one, i.e., $A_4=\{1,\ldots,r\}\cup\{w\}$, where each action $a_1 \in \{1,\dots,r\}$ is associated with variable $v_{a_1}$. The other players share the same set of actions $A$, with $A=A_1=A_2=A_3=\{\varphi_{ca} \mid c \in \{1, \dots, t\}, a \in \{1,\dots,8\}\}\cup\{\chi\}$, where each action $\varphi_{ca}$ is associated with one of the eight possible assignments of truth to the variables appearing in clause $\phi_c$, so that $\varphi_{ca}$ corresponds to the $a$-th assignment in the ordered set $L_{\phi_c}$. For each player $p \in \{1,2,3\}$, we define her utilities as follows:
  \begin{itemize}
  \item for each $a_4 \in A_4\setminus\{w\}$ and for each $a_1 \in A\setminus\{\chi\}$ with $a_1=\varphi_{ca}=l_1  l_2  l_3$, $U_p^{a_1 a_1 a_1 a_4}=1$ if $v(l_p)=v_{a_4}$ and $l_p$ is a positive literal or $v(l_p) \neq v_{a_4}$ and $l_p$ is negative;
  \item for each $a_4 \in A_4\setminus\{w\}$ and for each $a_1 \in A\setminus\{\chi\}$ with $a_1=\varphi_{ca}=l_1  l_2  l_3$, $U_p^{a_1 a_1 a_1 a_4}=0$ if $v(l_p)=v_{a_4}$ and $l_p$ is a negative literal or $v(l_p) \neq v_{a_4}$ and $l_p$ is positive;
  \item for each $a_1 \in A\setminus\{\chi\}$ with $a_1=\varphi_{ca}=l_1  l_2  l_3$, $U_p^{a_1 a_1 a_1 w}=0$ if $l_p$ is a positive literal, while $U_p^{a_1 a_1 a_1 w}=1$ otherwise;
  \item for each $a_4 \in A_4$, for each $a_1,a_2,a_3 \in A\setminus\{\chi\}$ with $\neq(a_1,a_2,a_3)$, $U_p^{a_1 a_2 a_3 a_4}=\frac{1}{r+2}$;
  \item for each $a_4 \in A_4$, $a_3 \in A\setminus\{\chi\}$, and $a_2 \in A\setminus\{\chi\}$ with $a_2 =\varphi_{ca}=l_1 l_2 l_3$, $U_1^{\chi a_2 a_3 a_4}=\frac{1}{r+1}$ if $l_1$ is a positive literal, whereas $U_1^{\chi a_2 a_3 a_4}=\frac{r}{r+1}$ if $l_1$ is negative, while $U_2^{\chi a_2 a_3 a_4}=U_3^{\chi a_2 a_3 a_4}=0$; 
  \item for each $a_4 \in A_4$, $a_3 \in A\setminus\{\chi\}$, and $a_1 \in A\setminus\{\chi\}$ with $a_1 =\varphi_{ca}=l_1 l_2 l_3$, $U_2^{a_1 \chi a_3 a_4}=\frac{1}{r+1}$ if $l_2$ is a positive literal, whereas $U_2^{a_1 \chi a_3 a_4}=\frac{r}{r+1}$ if $l_2$ is negative, while $U_1^{a_1 \chi a_3 a_4}=1$ and $U_3^{a_1 \chi a_3 a_4}=0$;
  \item for each $a_4 \in A_4$, $a_1 \in A\setminus\{\chi\}$, and $a_2 \in A\setminus\{\chi\}$ with $a_2=\varphi_{ca}=l_1 l_2 l_3$, $U_3^{a_1 a_2 \chi a_4}=\frac{1}{r+1}$ if $l_3$ is a positive literal, whereas $U_3^{a_1 a_2 \chi a_4}=\frac{r}{r+1}$ if $l_3$ is negative, while $U_1^{a_1 a_2 \chi a_4}=0$ and $U_2^{a_1 a_2 \chi a_4}=1$;
  \item for each $a_4 \in A_4$, $U_1^{a_1 \chi \chi a_4}=U_3^{a_1 \chi \chi a_4}=1$ and $U_2^{a_1 \chi \chi a_4}=0$, for all $a_1 \in A\setminus\{\chi\}$;
  \item for each $a_4 \in A_4$, $U_1^{\chi a_2 \chi a_4}=1$ and $U_2^{\chi a_2 \chi a_4}=U_3^{\chi a_2 \chi a_4}=0$, for all $a_2 \in A\setminus\{\chi\}$;
  \item for each $a_4 \in A_4$, $U_1^{\chi \chi a_3 a_4}=U_3^{\chi \chi a_3 a_4}=0$ and $U_2^{\chi \chi a_3 a_4}=1$, for all $a_3 \in A$.
  \end{itemize}
  The payoff matrix of the fourth player is so defined:
  \begin{itemize}
  \item for each $a_4 \in A_4$ and for each $a_1 \in A\setminus\{\chi\}$ with $a_1=\varphi_{ca}=l_1  l_2  l_3$, $U_4^{a_1 a_1 a_1 a_4}=\epsilon$ if the truth assignment identified by $\varphi_{ca}$ makes $\phi_c$ false (i.e., whenever, for each $p \in \{1,2,3\}$, the clause $\phi_c$ contains the negation of $l_p$), while $U_4^{a_1 a_1 a_1 a_4}=1$ otherwise, where $\epsilon>0$;
  \item for each $a_4 \in A_4$ and for each $a_1,a_2,a_3 \in A$ with $\neq(a_1,a_2,a_3)$, with the addition of the triple $(\chi,\chi,\chi)$, $U_4^{a_1 a_2 a_3 a_4}=0$.
  \end{itemize}
\end{definition}

Games adhering to Definition~\ref{def:sat_gadget} have some useful properties, which we formally state in the following proposition.

\begin{proposition}\label{prop:sat_prop}
	Given a game $\Gamma(C,V)$ and an action $a_1 \in A \setminus \{\chi\}$, with $a_1=\varphi_{ca}=l_1 l_2 l_3$, the outcome $(a_1,a_1,a_1)$ is an NE of the followers' game whenever the leader commits to a strategy $x_4 \in \Delta_4$ such that:
	\begin{itemize}
		\item $x_4^{a_4} \geq \frac{1}{r+1}$ if $v(l_p)=v_{a_4}$ and $l_p$ is a positive literal, for some $p \in \{1,2,3\}$;
		\item $x_4^{a_4} \leq \frac{1}{r+1}$ if $v(l_p)=v_{a_4}$ and $l_p$ is a negative literal, for some $p \in \{1,2,3\}$;
		\item $x_4^{a_4}$ can be any if $v(l_p) \neq v_{a_4}$ for each $p \in \{1,2,3\}$.
	\end{itemize}
	All the other outcomes of the followers' game cannot be NEs, for any of the leader's commitments.
\end{proposition}

\begin{proof}
Observe that, in the outcomes not in $\{(a_1,a_1,a_1) : a_1 \in A \setminus \{\chi\} \}$, the followers' payoffs do not depend on the leader's strategy $x_4$. Thus, outcomes $(a_1,a_2,a_3)$, for every $a_1,a_2,a_3 \in A \setminus \{\chi\}$ with $\neq (a_1,a_2,a_3)$, cannot be NEs, as the first follower would deviate by playing action $\chi$,
obtaining a utility at least of $\frac{1}{r+1}$, instead of $\frac{1}{r+2}$. Also, outcomes $(\chi,a_2,a_3)$, for all $a_2, a_3 \in A \setminus \{\chi\}$, are not NEs, since the second follower is better off playing $\chi$ (as she gets $1 > 0$). Analogously, outcomes $(a_1,\chi,a_3)$ cannot be NEs, for all $a_1, a_3 \in A \setminus \{\chi\}$, as the third follower would deviate to $\chi$ (providing her with a utility of $1 > 0$). A similar argument also applies to outcomes $(\chi,\chi,a_3)$, for all $a_3 \in A$, as the first follower has an incentive to deviate by playing any action different from $\chi$. Moreover, outcomes $(a_1,\chi,\chi)$ are not NEs, for all $a_1 \in A \setminus \{\chi\}$, as the second follower would deviate to any other action (providing her with a utility of $1$). The same holds for outcomes $(a_1,a_2,\chi)$, for all $a_1, a_2 \in A \setminus \{\chi\}$, where
the first follower would deviate and play action $\chi$, and for outcomes $(\chi,a_2,\chi)$, for all $a_2 \in \setminus \{\chi\}$, where the second follower
would deviate and play $\chi$.

	Therefore, the only outcomes which can be NEs in the followers' game are those in $\{(a_1,a_1,a_1) : a_1 \in A \setminus \{\chi\} \}$. Suppose the leader commits to an arbitrary mixed strategy $x_4 \in \Delta_4$. The outcome $(a_1,a_1,a_1)$, for $a_1 \in A \setminus \{\chi\}$ with $a_1=\varphi_{ca}=l_1 l_2 l_3$, provides follower $p$, for any $p \in \{1,2,3\}$, with a utility of $u_p$, such that:
	\begin{itemize}
		\item $u_p=x_4^{a_4}$ if $v(l_p)=v_{a_4}$ and $l_p$ is a positive literal;
		\item $u_p=1-x_4^{a_4}$ if $v(l_p)=v_{a_4}$ and $l_p$ is a negative literal;
	\end{itemize}
	Clearly, $(a_1,a_1,a_1)$ is an NE if the following conditions hold:
	\begin{itemize}
		\item $u_p \geq \frac{1}{r+1}$ for each $p \in \{1,2,3\}$ such that $l_p$ is positive, as otherwise follower $p$ would deviate and play $\chi$;
		\item $u_p \geq \frac{r}{r+1}$ for each $p \in \{1,2,3\}$ such that $l_p$ is negative, as otherwise follower $p$ would deviate and play $\chi$;
	\end{itemize}
The claim is proven by these conditions, together with the definition of $u_p$. \qed
\end{proof}

The property stated in Proposition~\ref{prop:sat_prop} has an interesting interpretation if we look at the strategy space of the leader. In particular, given a game $\Gamma(C,V)$, the leader's strategy space $\Delta_4$ is partitioned according to the boundaries $x_4^{a_4} = \frac{1}{r+1}$, for $a_4 \in A_4 \setminus\{w\}$, by which $\Delta_4$ is split into $2^r$ regions, each corresponding to a possible truth assignment to the variables in $V$. Specifically, in the assignment corresponding to some region, variable $v_{a_4}$ takes values TRUE if $x_4^{a_4} \geq \frac{1}{r+1}$, while it takes value FALSE if $x_4^{a_4} \leq \frac{1}{r+1}$. Moreover, an outcome $(a_1,a_1,a_1)$, for $a_1 \in A \setminus \{\chi\}$ and $a_1=\varphi_{ca}$, is an NE of the followers' game only in the regions of the leader's strategy space whose corresponding truth assignment is compatible with the one represented by $\varphi_{ca}$. For instance, if $\varphi_{ca}=\bar v_1 v_2 v_3$, the corresponding outcome is an NE only if $x_4^1 \leq \frac{1}{r+1}$, $x_4^2 \geq \frac{1}{r+1}$ and $x_4^3 \geq \frac{1}{r+1}$
(with no further restrictions on the other probabilities).

In order to better understand how these games are built, let us make a simplified example, using 2-SAT instead of 3-SAT. Given an instance of 2-SAT
(a restriction of 3-SAT in which each clause can only contain two literals), we build $\Gamma(C,V)$ as in Definition~\ref{def:sat_gadget}, using only two followers instead of three. Consider, as an example, the instance of 2-SAT and its corresponding game in Figure~\ref{fig:game_2sat}. As one can easily see, the only outcomes which can be NEs in the followers' game are those where both followers play the same action, with the exception of $(\chi,\chi)$.\footnote{In this simple example, outcomes $(a_1,\chi)$, for all $a_1 \in A \setminus \{\chi\}$, are always NEs in the followers' game. 
They can nevertheless be ignored since, if we considered 3-SAT, they would not be NEs as the third follower would have incentive to deviate by playing action $\chi$.}
For instance, let us consider outcome $(v_1 \bar v_2, v_1 \bar v_2)$. Given the leader's strategy $x_4 \in \Delta_4$, the followers' payoffs in such outcome are $(x_4^1,1-x_4^2)$, and, therefore, both followers have no incentive to deviate from $(v_1 \bar v_2, v_1 \bar v_2)$ (that is, to play action $\chi$) only when $x_4^1 \geq \frac{1}{3}$ and $x_4^2 \leq \frac{1}{3}$, which are the constraints identifying those regions of $\Delta_4$ that correspond to truth assignments compatible with $v_1 \bar v_2$.

\begin{figure}[!htp]
	\centering
	{\renewcommand{\arraystretch}{2}
		\begin{tabular}{r|c|c|c|c|c|}
			\multicolumn{1}{r}{}
			&  \multicolumn{1}{c}{$v_1 v_2$}
			& \multicolumn{1}{c}{$v_1 \bar v_2$} 
			& \multicolumn{1}{c}{$\bar v_1 v_2$} 
			& \multicolumn{1}{c}{$\bar v_1 \bar v_2$} 
			& \multicolumn{1}{c}{$\chi$} \\
			\cline{2-6}
			$v_1 v_2$ & $1,0,1$ & $\frac{1}{4},\frac{1}{4},0$ & $\frac{1}{4},\frac{1}{4},0$ & $\frac{1}{4},\frac{1}{4},0$ & $1,\frac{1}{3},0$ \\
			\cline{2-6}
			$v_1 \bar v_2$ & $\frac{1}{4},\frac{1}{4},0$ & $1,1,1$ & $\frac{1}{4},\frac{1}{4},0$ & $\frac{1}{4},\frac{1}{4},0$ & $1,\frac{2}{3},0$ \\
			\cline{2-6}
			$\bar v_1 v_2$ & $\frac{1}{4},\frac{1}{4},0$ & $\frac{1}{4},\frac{1}{4},0$ & $0,0,1$ & $\frac{1}{4},\frac{1}{4},0$ & $1,\frac{1}{3},0$ \\
			\cline{2-6}
			$\bar v_1 \bar v_2$ & $\frac{1}{4},\frac{1}{4},0$ & $\frac{1}{4},\frac{1}{4},0$ & $\frac{1}{4},\frac{1}{4},0$ & $0,1,\epsilon$ & $1,\frac{2}{3},0$ \\
			\cline{2-6}
			$\chi$ & $\frac{1}{3},0,0$ & $\frac{1}{3},0,0$ & $\frac{2}{3},0,0$ & $\frac{2}{3},0,0$ & $0,1,0$ \\
			\cline{2-6}
			\multicolumn{1}{r}{}
			& \multicolumn{5}{c}{$ 1 $}
		\end{tabular}}
		{\renewcommand{\arraystretch}{2}
			\begin{tabular}{r|c|c|c|c|c|}
				\multicolumn{1}{r}{}
				&  \multicolumn{1}{c}{$v_1 v_2$}
				& \multicolumn{1}{c}{$v_1 \bar v_2$} 
				& \multicolumn{1}{c}{$\bar v_1 v_2$} 
				& \multicolumn{1}{c}{$\bar v_1 \bar v_2$} 
				& \multicolumn{1}{c}{$\chi$} \\
				\cline{2-6}
				$v_1 v_2$ & $0,1,1$ & $\frac{1}{4},\frac{1}{4},0$ & $\frac{1}{4},\frac{1}{4},0$ & $\frac{1}{4},\frac{1}{4},0$ & $1,\frac{1}{3},0$ \\
				\cline{2-6}
				$v_1 \bar v_2$ & $\frac{1}{4},\frac{1}{4},0$ & $0,0,1$ & $\frac{1}{4},\frac{1}{4},0$ & $\frac{1}{4},\frac{1}{4},0$ & $1,\frac{2}{3},0$ \\
				\cline{2-6}
				$\bar v_1 v_2$ & $\frac{1}{4},\frac{1}{4},0$ & $\frac{1}{4},\frac{1}{4},0$ & $1,1,1$ & $\frac{1}{4},\frac{1}{4},0$ & $1,\frac{1}{3},0$ \\
				\cline{2-6}
				$\bar v_1 \bar v_2$ & $\frac{1}{4},\frac{1}{4},0$ & $\frac{1}{4},\frac{1}{4},0$ & $\frac{1}{4},\frac{1}{4},0$ & $1,0,\epsilon$ & $1,\frac{2}{3},0$ \\
				\cline{2-6}
				$\chi$ & $\frac{1}{3},0,0$ & $\frac{1}{3},0,0$ & $\frac{2}{3},0,0$ & $\frac{2}{3},0,0$ & $0,1,0$ \\
				\cline{2-6}
				\multicolumn{1}{r}{}
				& \multicolumn{5}{c}{$ 2 $}
			\end{tabular}}
			{\renewcommand{\arraystretch}{2}
				\begin{tabular}{r|c|c|c|c|c|}
					\multicolumn{1}{r}{}
					&  \multicolumn{1}{c}{$v_1 v_2$}
					& \multicolumn{1}{c}{$v_1 \bar v_2$} 
					& \multicolumn{1}{c}{$\bar v_1 v_2$} 
					& \multicolumn{1}{c}{$\bar v_1 \bar v_2$} 
					& \multicolumn{1}{c}{$\chi$} \\
					\cline{2-6}
					$v_1 v_2$ & $0,0,1$ & $\frac{1}{4},\frac{1}{4},0$ & $\frac{1}{4},\frac{1}{4},0$ & $\frac{1}{4},\frac{1}{4},0$ & $1,\frac{1}{3},0$ \\
					\cline{2-6}
					$v_1 \bar v_2$ & $\frac{1}{4},\frac{1}{4},0$ & $0,1,1$ & $\frac{1}{4},\frac{1}{4},0$ & $\frac{1}{4},\frac{1}{4},0$ & $1,\frac{2}{3},0$ \\
					\cline{2-6}
					$\bar v_1 v_2$ & $\frac{1}{4},\frac{1}{4},0$ & $\frac{1}{4},\frac{1}{4},0$ & $1,0,1$ & $\frac{1}{4},\frac{1}{4},0$ & $1,\frac{1}{3},0$ \\
					\cline{2-6}
					$\bar v_1 \bar v_2$ & $\frac{1}{4},\frac{1}{4},0$ & $\frac{1}{4},\frac{1}{4},0$ & $\frac{1}{4},\frac{1}{4},0$ & $0,1,\epsilon$ & $1,\frac{2}{3},0$ \\
					\cline{2-6}
					$\chi$ & $\frac{1}{3},0,0$ & $\frac{1}{3},0,0$ & $\frac{2}{3},0,0$ & $\frac{2}{3},0,0$ & $0,1,0$ \\
					\cline{2-6}
					\multicolumn{1}{r}{}
					& \multicolumn{5}{c}{$ w $}
				\end{tabular}}
		\caption{A $\Gamma(C,V)$ game corresponding to a 2-SAT instance with $V=\{v_1,v_2\}$ and $C=\{ \{ v_1 v_2 \} \}$. By choosing her action, the third player (the leader) selects one of the three matrices, while the first and the second players (the followers) select a row and a column, respectively.}
		\label{fig:game_2sat}
\end{figure}

We are now ready to state
the result.

\begin{theorem}
  Computing a P-LFPNE is not in Poly-\textsf{APX} even for $n=4$, unless \textsf{P} = \textsf{NP}.
\end{theorem}
\begin{proof}
	Given a generic 3-SAT instance, let us build its corresponding game $\Gamma(C,V)$, according to Definition~\ref{def:sat_gadget}. Clearly, this construction requires polynomial time, because $|A_4|=r+1$ and $|A|=|A_1|=|A_2|=|A_3|=8t+1$, which are polynomials in $r$ and $t$, and, therefore, the number of outcomes in $\Gamma(C,V)$ is polynomial in $r$ and $t$. Furthermore, let us
select $\epsilon \in \big(0,\frac{1}{2^{r}}\big)$
(the polynomiality of the reduction is preserved as $\frac{1}{2^{r}}$ is representable in binary with a polynomial number of bits).
	
By contradiction, let us assume that there exists a polynomial-time approximation algorithm $\mathcal{A}$ capable of constructing an approximate solution to the problem of computing a P-LFPNE with an approximation factor $\frac{1}{2^{r}}$. Observe that, if the 3-SAT instance is a YES instance (i.e., if it is feasible), there exists then a strategy $x_4 \in \Delta_4$ such that all the NEs of the resulting followers' game provide the leader with a utility of $1$, since there is a region corresponding to a truth assignment which makes all the clauses true. On the other hand, if the 3-SAT instance is a NO instance (i.e., if it is not satisfiable), in each region of the leader's strategy space there exits then an NE for the followers' game which provides the leader with a utility of $\epsilon$.
Due to the assumption of pessimism, the followers would, then, always play such equilibrium.
	
It follows that, when applied to $\Gamma(C,V)$, $\mathcal{A}$ would return an approximate solution with value greater than $\frac{1}{2^{r}}$ if and only if the 3-SAT instance is feasible.
Since this would provide us with a solution to 3-SAT in polynomial time, we conclude that P-LFPNE-s is not in Poly-\textsf{APX} unless \textsf{P} = \textsf{NP}. \qed
\end{proof}

\section{Single-Level Reformulation and Restriction}\label{sec:reform}

We propose, in this section, a single-level reformulation of the problem admitting a supremum but, in general, not a maximum, and a corresponding restriction which always admits optimal (restricted) solutions.

For notational simplicity, we consider, here, the case with $n=3$ players. The generalisation to $n \geq 3$ is, although notationally more involved, straightforward. With only two followers, Problem~\eqref{problem:pess}, i.e., the bilevel programming formulation we gave in Subsection~\ref{subsec:problem}, reads:

\everymath{\displaystyle}
\begin{equation} \label{problem:pessn3}
  \begin{array}{llllr}
  \sup_{x_3} \min_{x_1,x_2} & \multicolumn{3}{l}{\sum_{a_1 \in A_1} \sum_{a_2 \in A_2} \sum_{a_3 \in A_3} U_3^{a_1a_2a_3} x_1^{a_1} \, x_2^{a_2} \, x_3^{a_3}}\\
    \text{s.t.}           &  x_3 \in & \multicolumn{2}{l}{\Delta_3}\\
                          &  x_1 \in & \argmax_{x_1} & \sum_{a_1 \in A_1} \sum_{a_2 \in A_2} \sum_{a_3 \in A_3} U_1^{a_1a_2a_3} x_1^{a_1} \, x_2^{a_2} \, x_3^{a_3} \\
                          &          & \text{s.t.} & x_1 \in \Delta_1 \cap \{0,1\}^{m}\\
                          &  x_2 \in & \argmax_{x_2} & \sum_{a_1 \in A_1} \sum_{a_2 \in A_2} \sum_{a_3 \in A_3} U_2^{a_1a_2a_3} x_1^{a_1} \, x_2^{a_2} \, x_3^{a_3} \\
                          &          & \text{s.t.} & x_2 \in \Delta_2 \cap \{0,1\}^{m}.
  \end{array}
\end{equation} 
\everymath{\textstyle}

\vspace{-0.3cm}
\subsection{Single-Level Reformulation}

In order to cast Problem~\eqref{problem:pessn3} into a single-level problem, we introduce, first, a reformulation of the followers' problem:
\begin{lemma}
The following MILP, parametric in $x_3$, is an exact reformulation of the followers' problem of, given a leader's strategy $x_3$, finding a pure NE which minimises the leader's utility:
\begin{subequations}\label{prob:primal}
\begin{align}
\label{prob:primal1}  \min_{y}
  & \sum_{a_1 \in A_1} \sum_{a_2 \in A_2} y^{a_1a_2} \sum_{a_3 \in A_3} U_3^{a_1 a_2 a_3} x_3^{a_3} \\
\label{prob:primal2}  \text{\em s.t.}
  & \sum_{a_1 \in A_1} \sum_{a_2 \in A_2} y^{a_1a_2} =1\\
\label{prob:primal3} 
  & y^{a_1a_2} \sum_{a_3 \in A_3} (U_1^{a_1a_2a_3} - U_1^{a_1'a_2a_3}) x_3^{a_3}  \geq 0 & \forall a_1\in A_1, a_2 \in A_2, a_1' \in A_1\\
\label{prob:primal4}
  & y^{a_1a_2} \sum_{a_3 \in A_3} (U_2^{a_1a_2a_3}  -  U_2^{a_1a_2'a_3}) x_3^{a_3} \geq 0 & \forall a_1\in A_1, a_2 \in A_2, a_2' \in A_2\\
\label{prob:primal5}
  & y^{a_1a_2} \in \mathbb{Z}_+ & \forall a_1 \in A_1, a_2 \in A_2.
\end{align}
\end{subequations}
\end{lemma}
\begin{proof}
Note that, in Problem~\eqref{problem:pessn3}, a solution to the followers' problem satisfies $x_1^{a_1}=x_2^{a_2}=1$ for some $(a_1,a_2) \in A_1 \times A_2$ and $x_1^{a_1'}=x_2^{a_2'}=0$ for all $(a_1',a_2') \neq (a_1,a_2)$. Problem~\eqref{prob:primal} encodes this in terms of the variable $y^{a_1a_2}$ by imposing $y^{a_1a_2} = 1$ if an only if $(a_1,a_2)$ is a pessimistic NE. Let us look at this in detail.

Due to Constraints~\eqref{prob:primal2} and~\eqref{prob:primal5}, $y^{a_1a_2}$ is equal to 1 for one and only pair $(a_1,a_2)$. 

Due to Constraints~\eqref{prob:primal3} and~\eqref{prob:primal4}, for all $(a_1,a_2)$ such that $y^{a_1a_2}=1$, there can be no action $a_1' \in A_1$ (respectively, $a_2' \in A_2$) by which follower~1 (respectively, follower~2) could obtain a better payoff, assuming that the other follower would play action~$a_2$ (respectively, action $a_1$). This guarantees that $(a_1,a_2)$ be an NE.
Also note that Constraints~\eqref{prob:primal3} and~\eqref{prob:primal4} boil down to the tautology $0 \geq 0$ for any $(a_1,a_2) \in A_1 \times A_2$ with $y^{a_1a_2} = 0$.

By minimising the objective function (corresponding to the leader's utility), a pessimistic pure NE is found.\qed
\end{proof}

To arrive at a single-level reformulation of Problem~\eqref{problem:pessn3}, we rely on linear programming duality to restate Problem~\eqref{prob:primal} in terms of optimality conditions which do not employ the {\em min} operator. First, we show the following:
\begin{lemma}\label{lemma:tum}
The linear programming relaxation of Problem~\eqref{prob:primal} is integer.
\end{lemma}
\begin{proof}
Let us focus on Constraints~\eqref{prob:primal3} and analyze, for all $(a_1,a_2) \in A_1 \times A_2$ and $a_1' \in A_1$, the coefficient $\sum_{a_3 \in A_3} (U_1^{a_1a_2a_3} - U_1^{a_1'a_2a_3}) x_3^{a_3}$ which multiplies $y^{a_1a_2}$. The coefficient is equal to the {\em regret} player 1 would suffer from not playing action $a_1'$. If equal to 0, we have the tautology $0 \geq 0$. If $> 0$, we obtain, after dividing by $\sum_{a_3 \in A_3} (U_1^{a_1a_2a_3} - U_1^{a_1'a_2a_3}) x_3^{a_3}$ both sides of the constraint, $y^{a_1a_2} \geq 0$, which is subsumed by the nonnegativity of $y^{a_1a_2}$. If $< 0$, we obtain, after diving both sides of the constraint again by $\sum_{a_3 \in A_3} (U_1^{a_1a_2a_3} - U_1^{a_1'a_2a_3}) x_3^{a_3}$, $y^{a_1a_2} \leq 0$, which implies $y^{a_1a_2}=0$. A similar reasoning applies to Constraints~\eqref{prob:primal4}.

Let us now define $O$ as the set of pairs $(a_1,a_2)$ such that there is as least an action $a_1'$ or $a_2'$ for which one of the followers suffers from a strictly negative regret. We have $O := \{(a_1,a_2) \in A_1 \times A_2: \exists a_1' \in A_1 \text{ with } \sum_{a_3 \in A_3} (U_1^{a_1a_2a_3} - U_1^{a_1'a_2a_3}) x_3^{a_3} < 0 \vee \exists a_2' \in A_2 \text{ with } \sum_{a_3 \in A_3} (U_1^{a_1a_2a_3} - U_1^{a_1a_2'a_3}) x_3^{a_3} < 0\}$.

Relying on $O$, Problem~\eqref{prob:primal} can be rewritten as:
\begin{align*}
\min_y
  & \sum_{a_1 \in A_1} \sum_{a_2 \in A_2} y^{a_1a_2} \sum_{a_3 \in A_3} U_3^{a_1a_2a_3} x_3^{a_3}\\
\text{s.t.}
  & \sum_{a_1 \in A_1} \sum_{a_2 \in A_2} y^{a_1a_2} = 1\\
  & y^{a_1a_2} = 0 & \forall (a_1,a_2) \in O\\
  & y^{a_1a_2} \in \mathbb{Z}_+ & \forall a_1 \in A_1, a_2 \in A_2.
\end{align*}
%
Since, after discarding all variables $y^{a_1a_2}$ with $(a_1,a_2) \in O$, we obtain a problem with a single all-one constraint (whose constraint matrix is, therefore, totally unimodular), the integrality constraints can be dropped.\qed
\end{proof}

As a consequence of Lemma~\ref{lemma:tum}, the following can, finally, be established:
\begin{theorem}
The following single-level Quadratically Constrained Quadratic Program (QCQP) is an exact reformulation of Problem~\eqref{problem:pessn3}:

\begin{subequations}\label{prob:reform}
\small
\begin{align}
\label{prob:reform1}
\sup_{\substack{x_3,y\\\beta_1,\beta_2}} \quad
  & \sum_{a_1 \in A_1} \sum_{a_2 \in A_2} y^{a_1a_2}\sum_{a_3 \in A_3} U_3^{a_1a_2a_3} x_3^{a_3} \\
\label{prob:reform2} \text{\em s.t.} \quad
  & \sum_{a_1 \in A_1} \sum_{a_2 \in A_2} y^{a_1a_2} = 1\\
\label{prob:reform3}
  & y^{a_1a_2} \sum_{a_3 \in A_3} (U_1^{a_1a_2a_3} - U_1^{a_1'a_2a_3}) x_3^{a_3} \geq 0 & \forall a_1\in A_1, a_2 \in A_2, a_1' \in A_1\\
\label{prob:reform4}
  & y^{a_1a_2} \sum_{a_3 \in A_3} (U_2^{a_1a_2a_3}  -  U_2^{a_1a_2'a_3}) x_3^{a_3} \geq 0 & \forall a_1\in A_1, a_2 \in A_2, a_2' \in A_2\\
\nonumber
  & \sum_{a_1 \in A_1} \sum_{a_2 \in A_2} y^{a_1a_2} \sum_{a_3 \in A_3} U_3^{a_1a_2a_3} x_3^{a_3} \leq \sum_{a_3 \in A_3} U_3^{a_1a_2a_3} x_3^{a_3} + \hspace{-5cm}\\
\nonumber  & - \sum_{a_1' \in A_1} \beta_1^{a_1a_2a_1'}  \sum_{a_3 \in A_3} (U_1^{a_1a_2a_3} - U_1^{a_1'a_2a_3}) x_3^{a_3} + \hspace{-5cm}\\
\label{prob:reform5}
  & - \sum_{a_2' \in A_2} \beta_2^{a_1a_2a_2'} \sum_{a_3 \in A_3} (U_2^{a_1a_2a_3} - U_2^{a_1a_2'a_3}) x_3^{a_3} \hspace{-.5cm} & \forall a_1 \in A_1, a_2 \in A_2\\
\label{prob:reform5bis}
  & \sum_{a_3 \in A_3} x_3 = 1\\
\label{prob:reform6}
  & \beta_1^{a_1a_2a_1'} \geq 0 & \forall a_1 \in A_1, a_2 \in A_2, a_1' \in A_1\\
\label{prob:reform7}
  & \beta_2^{a_1a_2a_2'} \geq 0 & \forall a_1 \in A_1, a_2 \in A_2, a_2' \in A_2\\
\label{prob:reform8}
  & y^{a_1a_2} \geq 0 & \forall a_1 \in A_1, a_2 \in A_2\\
\label{prob:reform9}
  & x_3^{a_3} \geq 0 & \forall a_3 \in A_3.
\end{align}
\end{subequations}
\end{theorem}
\begin{proof}
First, by relying on Lemma~\ref{lemma:tum}, we introduce the linear programming dual of the linear programming relaxation of Problem~\eqref{prob:primal}. Letting $\alpha$, $\beta_1^{a_1a_2a_1'}$, and $\beta_2^{a_1a_2a_2'}$ be the dual variables of, respectively, Constraints~\eqref{prob:primal2},~\eqref{prob:primal3}, and~\eqref{prob:primal4}, the dual reads:
\begin{subequations}
\begin{align*}
\max_{\alpha,\beta_1,\beta_2} \;
  & \alpha \\
\text{s.t.} \;
  & \alpha +  \sum_{a_1' \in A_1} \beta_1^{a_1a_2a_1'} \sum_{a_3 \in A_3} (U_1^{a_1a_2a_3} - U_1^{a_1'a_2a_3}) x_3^{a_3}  + \hspace{-2cm}\\
  & \displaystyle+ \sum_{a_2' \in A_2} \beta_2^{a_1a_2a_2'} \sum_{a_3 \in A_3} (U_2^{a_1a_2a_3} - U_2^{a_1a_2'a_3}) x_3^{a_3} \hspace{-2cm}\\
  & \leq \sum_{a_3 \in A_3} U_3^{a_1a_2a_3} x_3^{a_3} & \; \forall a_1 \in A_1, a_2 \in A_2\\
  & \alpha \geq 0\\
  & \beta_1^{a_1a_2a_1'} \geq 0 & \forall a_1 \in A_1, a_2 \in A_2, a_1' \in A_1\\
  & \beta_2^{a_1a_2a_2'} \geq 0 & \forall a_1 \in A_1, a_2 \in A_2, a_2' \in A_2.
\end{align*}
\end{subequations}
A set of optimality conditions for Problem~\eqref{prob:primal} can then be derived by simultaneously imposing primal and dual feasibility for the sets of primal and dual variables (by imposing the respective constraints) and equating the objective functions of the two problems.

The dual variable $\alpha$ can be projected out from the resulting formulation via Fourier-Motzkin elimination, leading to Constraints~\eqref{prob:reform5}.

The result in the claim is obtained after introducing the leader's utility as the objective function of the problem and then casting the problem as a maximisation problem (in which a supremum is sought). \qed
\end{proof}

\subsection{Unboundedness}


Let us provide an interpretation of Problem~\eqref{prob:reform} from a purely game-theoretical perspective. As the left-hand side of each instance of Constraints~\eqref{prob:reform5} is equal to the leader's utility function (which is maximised), Constraints~\eqref{prob:reform5} account for the {\em maximin} aspect of the problem, imposing that the leader's utility be nonlarger than the utility she could obtain in any of the NEs arising in the followers' game. Observe that, for each $(a_1,a_2) \in A_1 \times A_2$, if $\exists a_1' \in A_1:  \sum_{a_3 \in A_3} (U_1^{a_1a_2a_3} - U_1^{a_1'a_2a_3}) x_3^{a_3} < 0$ (and, thus, $(a_1,a_2)$ is not an NE), the corresponding constraint can be trivially satisfied by letting $\beta_1^{a_1a_2a_1'} = \infty$. Similarly, for each $(a_1,a_2) \in A_1 \times A_2$, if $\exists a_2' \in A_2:  \sum_{a_3 \in A_3} (U_2^{a_1a_2a_3} - U_2^{a_1a_2'a_3}) x_3^{a_3} < 0$  (and, thus, $(a_1,a_2)$ is not an NE), the corresponding constraint can be trivially satisfied by letting $\beta_2^{a_1a_2a_2'} = \infty$. Differently, if the two aforementioned coefficients are $\geq 0$ for all $a_1' \in A_1$ and for all $a_2' \in A_2$, then, w.l.o.g., $\beta_1^{a_1a_2a_1'}=\beta_2^{a_1a_2a_2'} = 0$, as this corresponds to imposing the smallest restriction on the leader's utility. In this case, in particular, the instance of Constraint~\eqref{prob:reform5} amounts to $\sum_{a_1 \in A_1} \sum_{a_2 \in A_2} y^{a_1a_2} \sum_{a_3 \in A_3} U_3^{a_1a_2a_3} x_3^{a_3} \leq \sum_{a_3 \in A_3} U_3^{a_1a_2a_3} x_3^{a_3}$, thus imposing that the leader's utility be nonlarger than the value she should obtain at the NE corresponding to pair $(a_1,a_2)$.\footnote{Let us note that the primal part of Problem~\eqref{prob:reform} is necessary to obtain a bounded objective function value. Indeed, a purely dual formulation where a dummy variable $\eta$ is maximised, subject to being upper bounded by all the right-hand sides in Constraints~\eqref{prob:reform5}, but containing no primal part, would only be correct in case the followers' game admitted a pure NE {\em for every} $x_3 \in \Delta_3$. Indeed, if this were not the case, any $x_3$ for which the followers' game admitted no pure NE would lead to positive coefficients in each of the right-hand sides of Constraints~\eqref{prob:reform5}. Thus, by letting $\beta_1^{a_1a_2a_1'}=\beta_2^{a_1a_2a_2'} = \infty$ for all $a_1 \in A_1, a_2 \in A_2, a_1' \in A_1, a_2' \in A_2$, $\eta$ would be unbounded. This formulation would, though, be correct for a variant of the problem tackled here where the leader looks for, primarily, a strategy $x_3$ such that the followers' game admits no NE and, only if this is not possible, for a strategy maximising her utility in the worst-case NE.} 


In line with the results in Proposition~\ref{prop:nonexistence} of Subsection~\ref{subsec:preliminary}, the following holds:
\begin{proposition}\label{prop:unbounded}
In the general case, Problem~\eqref{prob:reform} admits a supremum but not a maximum. In particular, there is, in general, no finite upper bound on the dual variables $\beta_1^{a_1a_2a_1'}$ and $\beta_2^{a_1a_2a_2'}$ whose introduction would not restrict the set of optimal solutions of the problem.
\end{proposition}
\begin{proof}
Since the leader's objective function is bounded, Problem~\eqref{prob:reform} clearly admits a supremum.
To show that, in the general case, the problem does not admit a maximum over the reals, we show that there is at least a game instance which admits a sequence of feasible solutions whose map under the objective function converges to the supremum, while the series itself converges to a point where the objective function value is strictly smaller than the supremum itself. Consider once again the game introduced in the proof of Proposition~\ref{prop:nonexistence}. For this game, letting $x_3 = (1-\rho, \rho)$ for some $\rho \in [0,1]$, Constraints~\eqref{prob:reform5} read for, in order, $(a_1,a_2)=(1,1)$, $(a_1,a_2)=(1,2)$, $(a_1,a_2)=(2,1)$, and $(a_1,a_2)=(2,2)$, as follows:
\begin{align*}
y^{11}(0) +y^{12} (5+5\rho) + y^{21} (1) + y^{22} (0) & \leq 0 - \beta_1^{111}(0.5-\rho) - \beta_2^{112} (-1-\rho)\\
y^{11}(0) +y^{12} (5+5\rho) + y^{21} (1) + y^{22} (0) & \leq 5+5\rho - \beta_1^{122}(1+\rho) - \beta_2^{121} (1+\rho)\\
y^{11}(0) +y^{12} (5+5\rho) + y^{21} (1) + y^{22} (0) & \leq 1 - \beta_1^{211}(-0.5+\rho) - \beta_2^{212} (-0.5+\rho)\\
y^{11}(0) +y^{12} (5+5\rho) + y^{21} (1) + y^{22} (0) & \leq 0 - \beta_1^{221}(-1-\rho) - \beta_2^{221} (0.5-\rho).
\end{align*}

As discussed in the proof of Proposition~\ref{prop:nonexistence}, this game admits two NEs when $\rho=0.5$, namely, $(1,2)$ and $(2,1)$. Letting $y^{11} = y^{22} = 0$ and $\rho = 0.5$, Constraints~\eqref{prob:reform5} become:
\begin{align*}
7.5 y^{12} + y^{21} & \leq 0 + 1.5 \beta_2^{112}\\
7.5 y^{12} + y^{21} & \leq 7.5 - 1.5 \beta_1^{122} - 1.5\beta_2^{121}\\
7.5 y^{12} + y^{21} & \leq 1\\
7.5 y^{12} + y^{21} & \leq 0 + 1.5 \beta_1^{221}.
\end{align*}
As the left-hand side of each inequality corresponds to the objective function value (to be maximised), we look for nonnegative values for variables $\beta_2^{112},\beta_1^{112},\beta_2^{121},\beta_1^{221}$ for which the right-hand sides are as large as possible. W.l.o.g., we have $\beta_1^{122} = \beta_2^{121} = 0$ for the variables in the second constraint, as the coefficients multiplying them are negative.
As to the other two variables contained in the first and fourth constraints, it suffices to let $\beta_2^{112} = \beta_2^{221}  = 5$,
as, this way, the two right-hand sides become equal to $7.5$, thus only implying a trivially valid upper bound on the objective function value (as, due to $y^{12} + y^{21} = 1$, $7.5 y^{12} + y^{21} \leq 7.5$ holds). Overall, the four constraints impose $7.5 y^{12} + y^{21} \leq \min\{7.5, 1\} = 1$, which leads to a unique optimal solution of value 1 with $y^{21} = 1$.

Let now $\rho=0.5-\epsilon$, for some $\epsilon \in \left(0,0.5\right]$. Recall that, as shown in the proof of Proposition~\ref{prop:nonexistence}, $(1,2)$ is the unique pure NE in the followers' game for $\epsilon < 0.5$, while, for $\epsilon = 0.5$, the game admits two NEs: $(1,2)$ and $(2,1)$. Letting $y^{11} = y^{22} = 0$ and $\rho=0.5-\epsilon$, Constraints~\eqref{prob:reform5} read:
\begin{align*}
(7.5 - 5\epsilon) y^{12} + y^{21} & \leq 0 -\epsilon \beta_1^{112} + \left(1.5-\epsilon\right) \beta_2^{111}\\
(7.5 - 5\epsilon) y^{12} + y^{21} & \leq 7.5 -5 \epsilon - \left(1.5-\epsilon\right)\beta_1^{122} - \left(1.5-\epsilon\right)\beta_2^{121}\\
(7.5 - 5\epsilon) y^{12} + y^{21} & \leq 1 + \epsilon \beta_1^{211} + \epsilon \beta_2^{212}\\
(7.5 - 5\epsilon) y^{12} + y^{21} & \leq 0 + \left(1.5-\epsilon\right) \beta_1^{221} - \epsilon\beta_2^{221}.
\end{align*}
As $\epsilon \leq 0.5$, the coefficients of $\beta_1^{112}, \beta_1^{122}, \beta_2^{121}$, and $\beta_2^{221}$ are all negative and, thus, we can let $\beta_1^{112} = \beta_1^{122} = \beta_2^{121} = \beta_2^{221} = 0$, w.l.o.g.. As to variables $\beta_2^{111}$ and $\beta_1^{221}$, it suffices to impose $\beta_2^{111} = \beta_1^{221} = \frac{7.5 - 5\epsilon}{1.5-\epsilon}$, which corresponds to bounding the objective function by $7.5 - 5\epsilon$---this is a trivially valid bound as, since $y^{12}$ and $y^{21}$ are binary, $(7.5 - 5\epsilon) y^{12} + y^{21} \leq 7.5 - \epsilon$. As to variables $\beta_1^{211}$ and $\beta_2^{212}$, it suffices to choose any pair of values satisfying $\beta_1^{211} + \beta_2^{212} = \frac{7.5 - 5\epsilon}{\epsilon}$, thanks to which the corresponding constraint becomes $(7.5 - 5\epsilon) y^{12} + y^{21} \leq 7.5 - 5\epsilon$. With these choices, the four constraints boil down to $(7.5 - 5\epsilon) y^{12} + y^{21} \leq 7.5 - 5 \epsilon$, thus leading to an objective function value of $7.5 - 5 \epsilon$.

As it is clear, $7.5 - 5 \epsilon \rightarrow 7.5$ when $\epsilon \rightarrow 0$, which is in contrast to what we previously derived, namely, that the leader's utility is equal to 1 for $\epsilon = 0$. This shows that 7.5 is a supremum, but not a maximum.

The unboundedness of the dual variables $\beta_1^{a_1a_2a_1'}$ and $\beta_2^{a_1a_2a_2'}$ follows by noting that, for $\epsilon \rightarrow 0$, $\beta_1^{211} + \beta_2^{212} = \frac{7.5 - 5\epsilon}{\epsilon} \rightarrow \infty$. \qed
\end{proof}

\subsection{A Restricted Single-Level (MILP) Formulation}


We consider, here, the option of introducing an upper bound of $M$ on both $\beta_1^{a_1a_2a_1'}$ and $\beta_2^{a_1a_2a_2'}$, for all $a_1 \in A_1, a_2 \in A_2, a_1' \in A_1, a_2' \in A_2$. Due to the continuity of the objective function, this suffices to obtain a formulation which, although being a restriction of the original one, always admits a maximum (over the reals) as a consequence of Weierstrass' theorem. Quite conveniently, this restricted reformulation can be cast as an MILP, as we now show.
\begin{theorem}
The following MILP formulation is an exact reformulation of Problem~\eqref{prob:reform} for the case where $\beta_1^{a_1a_2a_1'} \leq M$ and $\beta_2^{a_1a_2a_2'} \leq M$ hold for all $a_1 \in A_1, a_2 \in A_2, a_1' \in A_1, a_2' \in A_2$, and a restricted one when these bounds are not valid:

\begin{subequations}\label{prob:reformMILP}
\small
\begin{align}
\label{prob:reformMILP1}  \max_{\substack{y, x_3,z\\q_1, q_2, p_1,p_2}} \; & \sum_{a_1 \in A_1} \sum_{a_2 \in A_2} \sum_{a_3 \in A_3} U_3^{a_1a_2a_3} z^{a_1a_2a_3} \\
\label{prob:reformMILP2}  \text{s.t.} \; & \sum_{a_1 \in A_1} \sum_{a_2 \in A_2} y^{a_1a_2} = 1\\
\label{prob:reformMILP3}  & \sum_{a_3 \in A_3} (U_1^{a_1a_2a_3} - U_1^{a_1'a_2a_3}) z^{a_1a_2a_3} \geq 0 \hspace{-2cm}& \forall a_1\in A_1, a_2 \in A_2, a_1' \in A_1\\
\label{prob:reformMILP4}  & \sum_{a_3 \in A_3} (U_2^{a_1a_2a_3}  -  U_2^{a_1a_2'a_3}) z^{a_1a_2a_3} \geq 0 \hspace{-2cm}& \forall a_1\in A_1, a_2 \in A_2, a_2' \in A_2\\
\nonumber  & \sum_{a_1 \in A_1} \sum_{a_2 \in A_2} \sum_{a_3 \in A_3} U_3^{a_1a_2a_3} z^{a_1a_2a_3} \leq \sum_{a_3 \in A_3} U_3^{a_1a_2a_3} x_3^{a_3} + \hspace{-20cm}\\
\nonumber & - M \sum_{a_1' \in A_1} \sum_{a_3 \in A_3} (U_1^{a_1a_2a_3} - U_1^{a_1'a_2a_3}) q_1^{a_1a_2a_1'a_3}  +  \hspace{-20cm}\\
\label{prob:reformMILP5} & - M \sum_{a_2' \in A_2}  \sum_{a_3 \in A_3} (U_2^{a_1a_2a_3} - U_2^{a_1a_2'a_3}) q_2^{a_1a_2a_2'a_3}  \hspace{-2.5cm} & \forall a_1 \in A_1, a_2 \in A_2\\
\label{prob:reformMILP5.1} & \sum_{a_3 \in A_3} x_3^{a_3} = 1\\
\label{prob:reformMILPmcz1} & z^{a_1a_2a_3} \geq x_3^{a_3} + y^{a_1a_2} - 1 & \forall a_1\in A_1, a_2 \in A_2, a_3 \in A_3\\
\label{prob:reformMILPmcz2} & z^{a_1a_2a_3} \leq x_3^{a_3}  & \forall a_1\in A_1, a_2 \in A_2, a_3 \in A_3\\
\label{prob:reformMILPmcz3} & z^{a_1a_2a_3} \leq y^{a_1a_2} & \forall a_1\in A_1, a_2 \in A_2, a_3 \in A_3\\
\label{prob:reformMILPmcq1}& q_1^{a_1a_2a_1'a_3} \geq x_3^{a_3} + p_1^{a_1a_2a_1'} - 1 & \forall a_1\in A_1, a_2 \in A_2, a_1' \in A_1, a_3 \in A_3\\
\label{prob:reformMILPmcq2}& q_1^{a_1a_2a_1'a_3} \leq x_3^{a_3}  & \forall a_1\in A_1, a_2 \in A_2, a_1' \in A_1, a_3 \in A_3\\
\label{prob:reformMILPmcq3}& q_1^{a_1a_2a_1'a_3} \leq p_1^{a_1a_2a_1'}  & \forall a_1\in A_1, a_2 \in A_2, a_1' \in A_1, a_3 \in A_3\\
\label{prob:reformMILPmcq4}& q_2^{a_1a_2a_2'a_3} \geq x_3^{a_3} + p_2^{a_1a_2a_2'} - 1  & \forall a_1\in A_1, a_2 \in A_2, a_2' \in A_2, a_3 \in A_3\\
\label{prob:reformMILPmcq5}& q_2^{a_1a_2a_2'a_3} \leq x_3^{a_3}  & \forall a_1\in A_1, a_2 \in A_2, a_2' \in A_2, a_3 \in A_3\\
\label{prob:reformMILPmcq6}& q_2^{a_1a_2a_2'a_3} \leq p_2^{a_1a_2a_2'} & \forall a_1\in A_1, a_2 \in A_2, a_2' \in A_2, a_3 \in A_3\\
\label{prob:reformMILP6.-1} & x_3^{a_3} \geq 0 & \forall a_3 \in A_3\\
\label{prob:reformMILP8}  & y^{a_1a_2} \in \{0,1\} & \forall a_1 \in A_1, a_2 \in A_2\\
\label{prob:reformMILP6}  & p_1^{a_1a_2a_1'} \in \{0,1\} & \forall a_1 \in A_1, a_2 \in A_2, a_1' \in A_1\\
\label{prob:reformMILP7}  & p_2^{a_1a_2a_2'} \in \{0,1\} & \forall a_1 \in A_1, a_2 \in A_2, a_2' \in A_2\\
\label{prob:reformMILP9} & z^{a_1a_2a_3} \geq 0 & \forall a_1\in A_1, a_2 \in A_2, a_3 \in A_3\\
\label{prob:reformMILP10}& q_1^{a_1a_2a_1'a_3} \geq 0 & \forall a_1\in A_1, a_2 \in A_2, a_1' \in A_1, a_3 \in A_3\\
\label{prob:reformMILP11}& q_2^{a_1a_2a_2'a_3} \geq 0 & \forall a_1\in A_1, a_2 \in A_2, a_2' \in A_2, a_3 \in A_3.
\end{align}
\end{subequations}
\end{theorem}
\begin{proof}
After introducing the variable $z^{a_1a_2a_3}$, each bilinear product $y^{a_1a_2} x_3^{a_3}$ in Problem~\eqref{prob:reform} can be linearised by substituting $z^{a_1a_2a_3}$ for it and introducing the McCormick Envelope Constraints~\eqref{prob:reformMILPmcz1}--\eqref{prob:reformMILPmcz3}, which are sufficient to guarantee $z^{a_1a_2a_3} = y^{a_1a_2} x_3^{a_3}$ if $y^{a_1a_2}$ takes binary values~\cite{mccormick1976computability}.
%
Assuming $\beta_1^{a_1a_2a_1'} \in [0,M]$ for each $a_1 \in A_1, a_2 \in A_2, a_1' \in A_1$, we can clearly restrict ourselves to $\beta_1^{a_1a_2a_1'} \in \{0,M\}$. We can, therefore, introduce the variable $p_1^{a_1a_2a_1'} \in \{0,1\}$, substituting $M p_1^{a_1a_2a_1'}$ for each occurrence of $\beta_1^{a_1a_2a_1'}$. This way, for each $a_1 \in A_1, a_2 \in A_2, a_1' \in A_1$, the term $\beta_1^{a_1a_2a_1'}  \sum_{a_3 \in A_3} (U_1^{a_1a_2a_3} - U_1^{a_1'a_2a_3}) x_3^{a_3}$ becomes $ M\sum_{a_3 \in A_3} (U_1^{a_1a_2a_3} - U_1^{a_1'a_2a_3}) p_1^{a_1a_2a_1'} x_3^{a_3}$. We can, then, introduce the variable $q_1^{a_1a_2a_1'a_3}$ and impose $q_1^{a_1a_2a_1'a_3} = p_1^{a_1a_2a_1'} x_3^{a_3}$ via the McCormick Envelope Constraints~\eqref{prob:reformMILPmcq1}--\eqref{prob:reformMILPmcq3}. This way, the term $ M\sum_{a_3 \in A_3} (U_1^{a_1a_2a_3} - U_1^{a_1'a_2a_3}) p_1^{a_1a_2a_1'} x_3^{a_3}$ becomes the completely linear term $M \sum_{a_1' \in A_1} \sum_{a_3 \in A_3} (U_1^{a_1a_2a_3} - U_1^{a_1'a_2a_3}) q_1^{a_1a_2a_1'a_3}$. Similar arguments hold for $\beta_2^{a_1a_2a_2'}$, leading to the introduction of Constraints~\eqref{prob:reformMILPmcq4}--\eqref{prob:reformMILPmcq6}. \qed
\end{proof}

The impact of bounding $\beta_1^{a_1a_2a_1'}$ and $\beta_2^{a_1a_2a_2'}$ by $M$ is explained as follows. Assume that those upper bounds are introduced into Problem~\eqref{prob:reform}. If $M$ is not large enough for the chosen $x_3$ (remember that, as shown in Proposition~\ref{prop:unbounded}, one may need $M \rightarrow \infty$ for $x_3$ approaching a discontinuity point of the leader's utility function), Constraints~\eqref{prob:reform5} may remain active for some $(\hat a_1,\hat a_2)$ which is not an NE for the chosen $x_3$. Let $(a_1,a_2)$ be the worst-case NE the followers would play and assume that the right-hand side of Constraint~\eqref{prob:reform5} for $(\hat a_1,\hat a_2)$ is strictly smaller than the utility the leader would obtain if the followers played the NE $(a_1,a_2)$, namely, $\sum_{a_3 \in A_3} U_3^{\hat a_1 \hat a_2 a_3} x_3^{a_3} - \sum_{a_1' \in A_1} \beta_1^{\hat a_1 \hat a_2 a_1'}  \sum_{a_3 \in A_3} (U_1^{\hat a_1 \hat a_2 a_3} - U_1^{a_1' \hat a_2 a_3}) x_3^{a_3}  - \sum_{a_2' \in A_2} \beta_2^{\hat a_1 \hat a_2 a_2'} \sum_{a_3 \in A_3} (U_2^{\hat a_1 \hat a_2 a_3} - U_2^{\hat a_1 a_2'a_3}) x_3^{a_3} < \sum_{a_3 \in A_3} U_3^{a_1a_2a_3} x_3^{a_3}$. Since, by letting $y^{a_1a_2} = 1$, the constraint would be violated (as, with that value of $y$, the left-hand side of the constraint would be $\sum_{a_3 \in A_3} U_3^{a_1a_2a_3} x_3^{a_3}$, which we assumed to be strictly larger than the right-hand side), this forces the choice of a different $x_3$ for which the upper bound of $M$ on $\beta_1^{a_1a_2a_1'}$ and $\beta_2^{a_1a_2a_2'}$ is sufficiently large not to cause the same issue with the worst-case NE corresponding to that $x_3$, thus restricting the set of strategies the leader could play.

In spite of this, by solving Problem~\eqref{prob:reform}, we are always guaranteed to find optimal (restricted) solutions to it (if $M$ is large enough for the restricted problem to admit feasible solutions). Such solutions correspond to feasible strategies of the leader, guaranteeing her a lower bound on her utility at a P-LFPNE.

\section{Exact Algorithm}\label{sec:exact_algorithm}

We propose, in this section, an exact exponential-time algorithm for the computation of a P-LFPNE, i.e., of $\sup_{x_n \in \Delta_n} f(x_n)$, which does not suffer from the shortcomings of the formulations we introduced in the previous section.
In particular, if there is no $x_n \in \Delta_n$ where the leader's utility $f(x_n)$ achieves $\sup_{x_n \in \Delta_n} f(x_n)$ (as $f(x_n)$ does not admit a maximum), our algorithm also returns, together with the supremum, a strategy $\hat x_n$ which provides the leader with a utility equal to an $\alpha$-approximation (in the additive sense) of the supremum, namely, a strategy $x_n$ satisfying $\sup_{x_n \in \Delta_n} f(x_n) - f(\hat x_n) \leq \alpha$, for any $\alpha>0$ chosen {\em a priori}. We first introduce, in Subsection~\ref{sub_sec:enum_alg}, a version of the algorithm based on explicit enumeration, which we then embed, in Subsection~\ref{sub_sec:bb_algorithm}, into a branch-and-bound scheme.

In the remainder of the section, we denote the closure of a set $X \subseteq \Delta_n$ relative to $\aff(\Delta_n)$ by $\overline{X}$, its boundary relative to $\aff(\Delta_n)$ by $\bd(X)$, and its complement relative to $\Delta_n$ by $X^c$. Note that, here, $\aff(\Delta_n)$ denotes the \emph{affine hull} of $\Delta_n$, i.e., the hyperplane in $\mathbb{R}^{m}$ containing $\Delta_n$.

\subsection{Enumerative Algorithm}\label{sub_sec:enum_alg}

\subsubsection{Computing $\sup_{x_n \in \Delta_n} f(x_n)$}


The key ingredient of our algorithm is what we call {\em outcome configurations}. Letting $A_F = \bigtimes_{p \in F} A_p$, we say that, for a given $x_n \in \Delta_n$, a pair $(S^+, S^-)$ with $S^+ \subseteq A_F$ and $S^- = A_F \setminus S^+$ is an outcome configuration if, in the followers' game induced by $x_n$, all the followers' action profiles $a_{-n} \in S^+$ constitute an NE and all the action profiles $a_{-n} \in S^-$ do not.

For every $a_{-n} \in A_F$, we define $X(a_{-n})$ as the set of all leader's strategies $x_n \in \Delta_n$ for which $a_{-n}$ is an NE in the followers' game induced by~$x_n$. Formally, $X(a_{-n})$ corresponds to the following (closed) polytope:
$$
X(a_{-n}) := \left\{ \begin{array}{lr}x_n \in \Delta_n: & \displaystyle \sum_{a_n \in A_n} U_p^{a_{-n}, a_n} x_n^{a_n} \geq \sum_{a_n \in A_n} U_p^{a_{-n}', a_n} x_n^{a_n} \; \forall p \in F, a_p' \in A_p \setminus \{a_p\}\\
& \text{with } a_{-n}'=(a_1,\ldots,a_{p-1},a_p',a_{p+1},\ldots,a_{n-1})  \end{array}\right\}.
$$

For each $a_{-n} \in A_F$, we also introduce the set $X^c(a_{-n})$ of all $x_n \in \Delta_n$ for which $a_{-n}$ is {\em not} an NE. For that purpose, we first define, for each $p \in F$, the following set:
$$ 
\displaystyle D_p(a_{-n},a_p') := \left\{\begin{array}{lr} x_n \in \Delta_n: & \displaystyle \sum_{a_n \in A_n} U_p^{a_{-n}, a_n} x_n^{a_n} < \sum_{a_n \in A_n} U_p^{a_{-n}', a_n} x_n^{a_n}\\
& \text{with } a_{-n}'=(a_1,\ldots,a_{p-1},a_p',a_{p+1},\ldots,a_{n-1})  \end{array}\right\}.
$$
$D_p(a_{-n},a_p')$, which is a not open nor closed polytope (as it has a missing facet, the one corresponding to its strict inequality), is the set of all values of~$x_n$ for which player $p$ would achieve a better utility by deviating from $a_{-n}$ and playing a different action $a_p' \in A_p$. Since, in principle, any player could deviate from $a_{-n}$ by playing any action not in $a_{-n}$, $X^c(a_{-n})$ is the following {\em disjunctive set}:
$$
\displaystyle X^c(a_{-n}) := \bigcup_{p \in F} \left( \bigcup_{a_p' \in A_p \setminus \{a_p\}} D_p(a_{-n},a_p') \right).
$$

Notice that, since any point in $\bd(X^c(a_{-n}))$ which is not in $\bd(\Delta_n)$ would satisfy, for some $a_p'$, the (strict, originally) inequality of $D_p(a_{-n},a_p')$ as an equation, such point is not in $X^c(a_{-n})$ and, hence, $\bd(X^c(a_{-n})) \cap X^c(a_{-n}) \subseteq \bd(\Delta_n)$. The closure $\overline{X^c(a_{-n})}$ of $X^c(a_{-n})$ is obtained by turning the strict constraint in the definition of each $D_p(a_{-n},a_p')$ into a nonstrict one. An illustration of $X(a_{-n})$ and $X^c(a_{-n})$, together with the closure $\overline{X^c(a_{-n})}$ of the latter, is reported in Figure~\ref{fig:simplices}.


\begin{figure}[h!]
\begin{center}
\begin{tabular}{ccc}
\hspace{-0.3cm}
\includegraphics[scale=0.3]{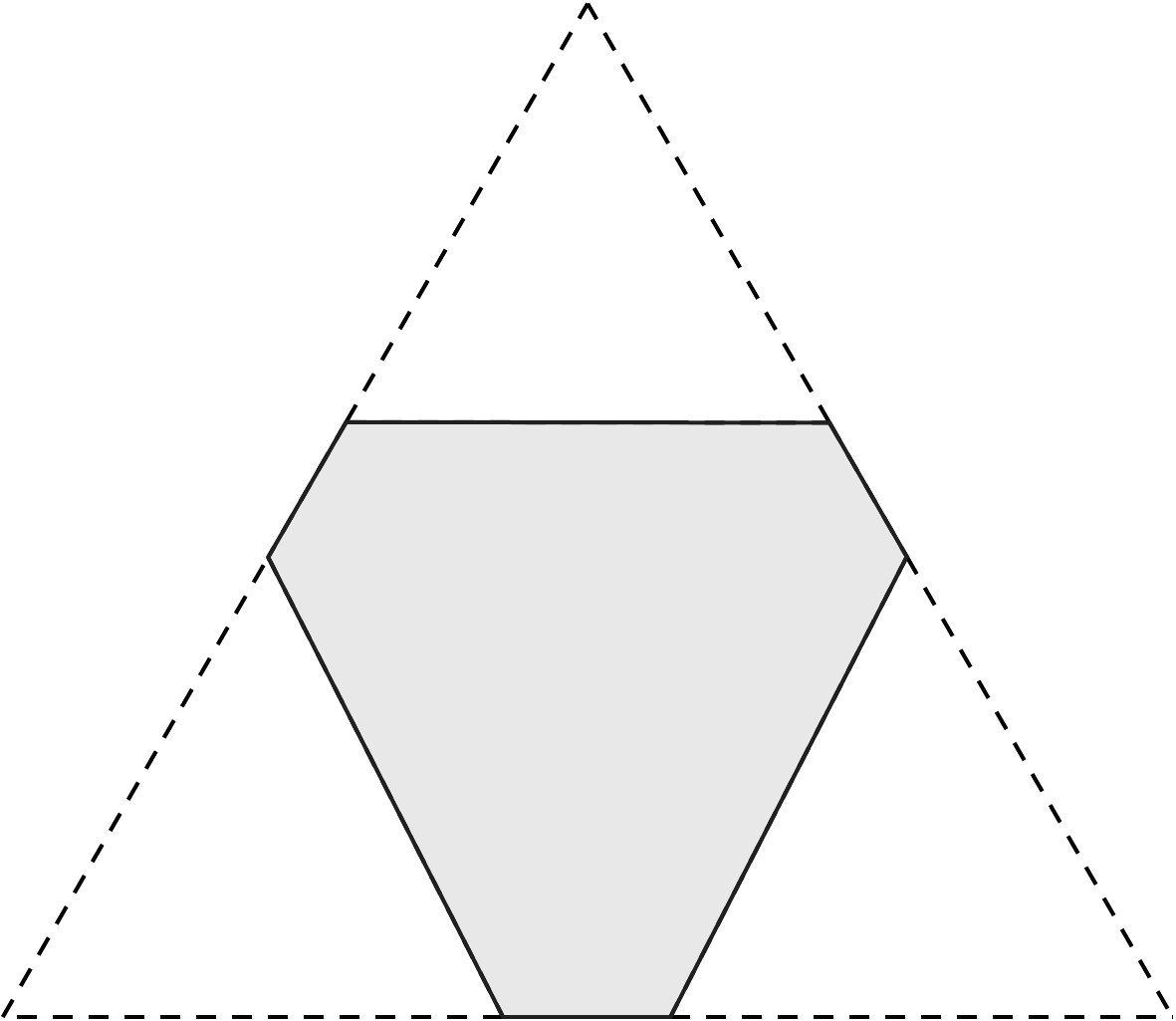} &
\includegraphics[scale=0.35]{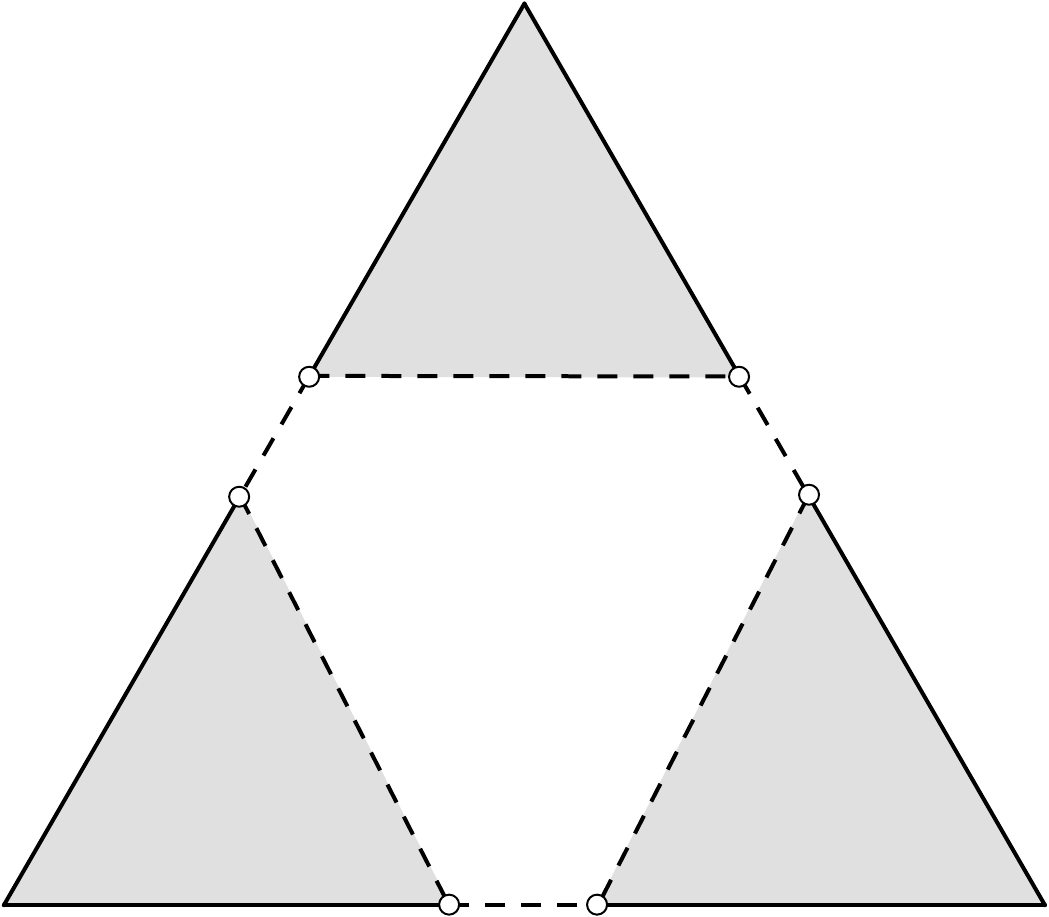} &
\includegraphics[scale=0.35]{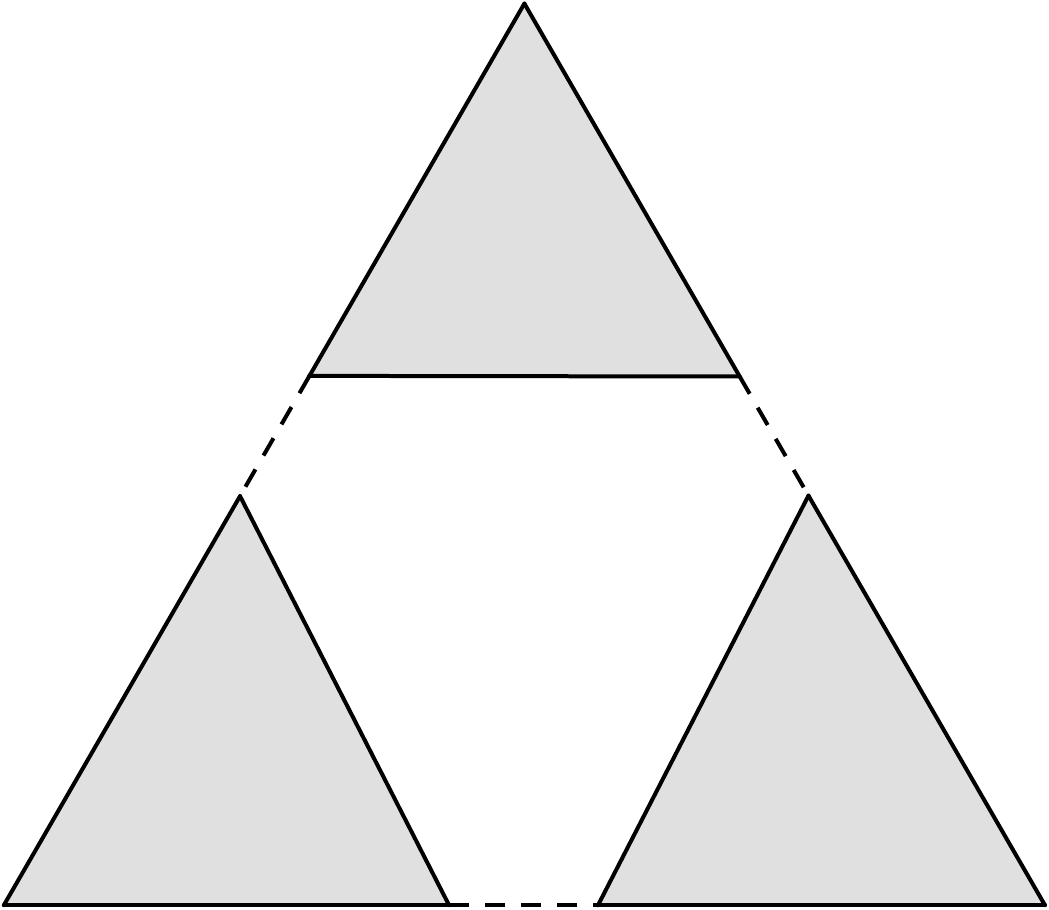}\\
$X(a_{-n})$ & $X^c(a_{-n})$ & $\overline{X^c(a_{-n})}$
\end{tabular}
\end{center}
\caption{An illustration of $X(a_{-n})$, $X^c(a_{-n})$, and $\overline{X^c(a_{-n})}$ for the case with $m=3$. The three sets are depicted as subsets (highlighted in gray and continuous lines) of the leader's strategy space $\Delta_n$. Dashed lines and circles indicate parts of $\Delta_n$ which are not contained in the sets.}
\label{fig:simplices}
\end{figure}

For every outcome configuration $(S^+,S^-)$, we introduce the following sets:
$$X(S^+) := \bigcap_{a_{-n}\in S^+} X(a_{-n})$$
and
$$X(S^-) := \bigcap_{a_{-n}\in S^-} X^c(a_{-n}).$$
While the former is a closed polytope, the latter is the union of not open nor closed polytopes and, thus, it is not open nor closed itself. Similarly to $X^c(a_{-n})$, $X(S^-)$ satisfies $\bd(X(S^-)) \cap X(S^-) \subseteq \bd(\Delta_n)$. The closure $\overline{X(S^-)}$ of $X(S^-)$ is obtained by taking the closure of each $X^c(a_{-n})$. Hence, $\overline{X(S^-)} = \bigcap_{a_{-n}\in S^-} \overline{X^c(a_{-n})}$.


By leveraging these definitions, we can now focus on the set of all leader's strategies which {\em realize} the outcome configuration $(S^+,S^-)$, namely:
$$X(S^+) \cap X(S^-).$$
As for $X(S^-)$, $X(S^+) \cap X(S^-)$ is not an open nor a closed set. Due to $X(S^+)$ being closed, the only points of $\bd(X(S^+) \cap X(S^-))$ which are not in $X(S^+) \cap X(S^-)$ itself are the very points in $\bd(X(S^-))$ which are not in $X(S^-)$. As a consequence,  $\overline{X(S^+) \cap X(S^-)} =  X(S^+) \cap \overline{X(S^-)}$.

Let us define the set $P := \{(S^+,S^-) : S^+ \in 2^{A_F} \wedge S^- = 2^{A_F} \setminus S^+\}$, which contains all the outcome configurations of the game. The following theorem highlights the structure of $f(x_n)$, suggesting an iterative way of expressing the problem of computing $\sup_{x_n \in \Delta_n} f(x_n)$. We will rely on it when designing our algorithm.
%
\begin{theorem}\label{thm:sup}
Let $\displaystyle \psi(x_n;S^+) := \min_{a_{-n} \in S^+} \sum_{a_n \in A_n} U^{a_{-n}, a_n} x_n^{a_n}$. The following holds:
\begin{equation*}
  \sup_{x_n \in \Delta_n} f(x_n) = \max_{\substack{(S^+,S^-) \in P: \\X(S^+)\cap X(S^-) \neq \emptyset}} \; \max_{x_n \in X(S^+) \cap \overline{X(S^-)}} \psi(x_n;S^+).
\end{equation*}
\end{theorem}
\begin{proof}
Let $\Delta'_n$ be the set of leader's strategies $x_n$ for which there exists a pure NE in the followers' game induced by $x_n$, namely, $\Delta'_n := \{x_n \in \Delta_n: f(x_n) > -\infty \}$. Since, by definition, $f(x_n) = -\infty$ for any $x_n \notin \Delta'_n$ and the supremum of $f(x_n)$ is finite due to the finiteness of the payoffs (and assuming the followers' game admits at least a pure NE for some $x_n \in \Delta_n$), we can, w.l.o.g., focus on $\Delta_n'$ and solve $\sup_{x_n \in \Delta_n'} f(x_n)$. In particular, the collection of the sets $X(S^+) \cap X(S^-) \neq \emptyset$ which are obtained for all $(S^+,S^-) \in P$ forms a partition of $\Delta_n'$. Due to the fact that, at any $x_n \in X(S^+) \cap X(S^-)$, the only pure NEs induced by $x_n$ in the followers' game are those in $S^+$, $f(x_n) = \psi(x_n;S^+)$. Since, as it is clear, the supremum of a function defined over a set is equal to the largest of the suprema of that function over the subsets of such set, we have:
\begin{equation*}
	\sup_{x_n \in \Delta_n} f(x_n) = \max_{\substack{(S^+,S^-) \in P: \\X(S^+)\cap X(S^-) \neq \emptyset}} \sup_{x_n \in X(S^+) \cap X(S^-)} \psi(x_n;S^+).
\end{equation*}

What remains to show is that, for all $X(S^+)\cap X(S^-) \neq \emptyset$, the following relationship holds:
$$
\sup_{x_n \in X(S^+)\cap X(S^-)} \psi(x_n;S^+) = \max_{x_n \in X(S^+) \cap \overline{ X(S^-)}} \psi(x_n;S^+).
$$
Since $\psi(x_n;S^+)$ is a continuous function (it is the point-wise minimum of finitely many continuous functions), its supremum over $X(S^+) \cap X(S^-)$ equals its maximum over the closure $\overline{X(S^+) \cap X(S^-)}$ of that set. Hence, the relationship follows due to $\overline{ X(S^+) \cap X(S^-)} = X(S^+) \cap \overline{ X(S^-)}$. \qed
\end{proof}

In particular, Theorem~\ref{thm:sup} shows that $f(x_n)$ is a piecewise function with a piece for each set $X(S^+) \cap X(S^-)$, each of which corresponding to the (continuous over its domain) piecewise-affine function $\psi(x_n;S^+)$. It follows that the only discontinuities of $f(x_n)$, due to which $f(x_n)$ may admit a supremum but not a maximum, are those where, in $\Delta_n$, $x_n$ transitions from a set $X(S^+) \cap X(S^-)$ to another one.


We show how to translate the formula in Theorem~\ref{thm:sup} into an algorithm by proving the following theorem:
\begin{theorem}\label{thm:alg_sup}
There exists a finite, exponential-time algorithm which computes $\sup_{x_n \in \Delta_n} f(x_n)$ and, whenever $\sup_{x_n \in \Delta_n} f(x_n) = \max_{x_n \in \Delta_n} f(x_n)$, also returns a strategy $x_n^*$ with $f(x_n^*) = \max_{x_n \in \Delta_n} f(x_n)$.
\end{theorem}
\begin{proof}
The algorithm relies on the expression given in Theorem~\ref{thm:sup}. All pairs $(S^+,S^-) \in P$ can be constructed by enumeration in time exponential in
the size of the instance.\footnote{Recall that the size of a game instance is $\Omega (m^n)$.} 
In particular, the set $P$ contains $2^{m^{n-1}}$ outcome configurations, each corresponding to a bi-partition of the outcomes of the followers' game into $S^+$ and $S^-$ (there are $m^{n-1}$ such outcomes, due to having $m$ actions and $n-1$ followers).

Let us define, for every $p \in F$, the following sets, parametric in $\epsilon \geq 0$:
$$ 
\displaystyle D_p(a_{-n},a_p';\epsilon) := \left\{\begin{array}{lr} x_n \in \Delta_n: & \displaystyle \sum_{a_n \in A_n} U_p^{a_{-n}, a_n} x_n^{a_n} + \epsilon \leq \sum_{a_n \in A_n} U_p^{a_{-n}', a_n} x_n^{a_n}\\
& \text{with } a_{-n}'=(a_1,\ldots,a_{p-1},a_p',a_{p+1},\ldots,a_{n-1})  \end{array}\right\},
$$
$$
\displaystyle X^c(a_{-n};\epsilon) := \bigcup_{p \in F} \left( \bigcup_{a_p' \in A_p \setminus \{a_p\}} D_p(a_{-n},a_p';\epsilon) \right),
$$
$$X(S^-;\epsilon) := \bigcap_{a_{-n}\in S^-} X^c(a_{-n};\epsilon).$$

Notice that, when $\epsilon = 0$, we have $D_p(a_{-n},a_p';0) = \overline{D_p(a_{-n},a_p')}$  and $X(S^-;0) = \overline{X(S^-)}$. Hence, we can verify whether $X(S^+) \cap X(S^-) \neq \emptyset$ by verifying whether there exists some $\epsilon >0$ such that $X(S^+) \cap X(S^-;\epsilon) \neq \emptyset$.
This can be done by solving the following problem and checking the strict positivity of $\epsilon$ in its solution:
\begin{equation}\label{prob:eps}
\begin{array}{ll}
\displaystyle
\max_{\epsilon,x_n} \quad & \epsilon\\
\text{s.t.}\quad & x_n \in X(S^+) \cap X(S^-;\epsilon)\\
                 & \epsilon \geq 0\\
                 & x_n \in \Delta_n.
\end{array}
\end{equation}
Problem~\eqref{prob:eps} can be cast as an MILP. To see this, observe that each $X^c(a_{-n};\epsilon)$ can be expressed as an MILP 
with a binary variable for each term of the disjunction which composes it, namely:
\begin{subequations}\label{cons:milp}
  \begin{align}
\nonumber
\nonumber    & \sum_{a_n \in A_n} U_p^{a_{-n}, a_n} x_n^{a_n} + \epsilon \leq \sum_{a_n \in A_n} U_p^{a_{-n}', a_n} x_n^{a_n} + M_p^{a_{-n},a_p'} z_p^{a_{-n},a_p'} \hspace{-2cm}\\
\label{cons:milp:1}
    & \hspace{.5cm}\forall p \in F, a_p' \in A_p \setminus \{a_p\}, \text{with } a_{-n}' = (a_1,\ldots,a_{p-1},a_p',a_{p+1},\ldots,a_{n-1}) \hspace{-4cm}\\
    & \sum_{p \in F} \sum_{a_p' \in A_p \setminus \{a_p\}} (1-z_p^{a_{-n}, a_p'}) = 1\\
    & z_p^{a_{-n}, a_p'} \in \{0,1\} & \forall p \in F, a_p' \in A_p \setminus \{a_p\}\\
    & x_n \in \Delta_n\\
    & \epsilon \geq 0.
  \end{align}
\end{subequations}
In Constraints~\eqref{cons:milp}, the constant $M_p^{a_{-n}, a_p'}$, which satisfies $M_p^{a_{-n}, a_p'} = \max_{a_n \in A_n} \{U_p^{a_{-n}, a_n} - U_p^{a_{-n}', a_n}\}$, is key to deactivate any instance of Constraints~\eqref{cons:milp:1} when the corresponding $z_p^{a_{-n},a_p'}$ is equal to 1.
%
The set $X(S^-;\epsilon)$ is obtained by simultaneously imposing Constraints~\eqref{cons:milp} for all $a_{-n} \in S^-$.

After verifying $X(S^+) \cap \overline{X(S^-)} \neq \emptyset$ by solving Problem~\eqref{prob:eps}, the value of $\max_{x_n \in X(S^+) \cap \overline{X(S^-)}} \psi(x_n;S^+)$ can be computed in, at most, exponential time by solving the following MILP:
\begin{equation}\label{prob:ex_lex_milp}
\everymath{\displaystyle}
\begin{array}{lll}
	\max_{\eta,x_n} & \eta\\
        \text{s.t.}         & \eta \leq \displaystyle \sum_{a_n \in A_n} U_n^{a_{-n}, a_n} x_n^{a_n} & \forall a_{-n} \in S^+\\
	                    & x_n \in X(S^+) \cap X(S^-;0)\\
                            & \eta \in \mathbb{R}\\
                            & x_n \in \Delta_n,
\end{array}
\end{equation}
where the first constraint accounts for the maxmin aspect of the problem.
The largest value of $\eta$ found over all sets $X(S^+) \cap X(S^-)$, for all $(S^+,S^-) \in P$, corresponds to $\sup_{x_n \in \Delta_n} f(x_n)$.

To verify whether $f(x_n)$ admits $\max_{x_n \in \Delta_n} f(x_n)$ and, if it does, compute it, we solve, in the algorithm, the following problem (rather than the aforementioned $\max_{x_n \in X(S^+) \cap \overline{X(S^-)}} \psi(x_n;S^+)$):
\begin{equation}\label{prob:lexolo}
\lexmax_{\epsilon \geq 0, x_n \in X(S^+) \cap X(S^-; \epsilon)} [\psi(x_n;S^+);\epsilon].
\end{equation}
This problem calls for a pair $(x_n,\epsilon)$ with $x_n \in X(S^+) \cap X(S^-; \epsilon)$ such that, among all pairs which maximise $\psi(x_n;S^+)$, $\epsilon$ is as large as possible. This way, in any solution $(x_n, \epsilon)$ with $\epsilon > 0$ we have $x_n \in X(S^+) \cap X(S^-)$ (rather than $x_n \in X(S^+) \cap \overline{X(S^-)}$). Since, there, $\psi(x_n;S^+) = f(x_n)$, we conclude that $f(x_n)$ admits a maximum (equal to the value of the supremum)
if $\epsilon > 0$, whereas it only admits a supremum if $\epsilon = 0$.

Problem~\eqref{prob:lexolo} can be solved in, at most, exponential time by solving the following lex-MILP:
\begin{equation}\label{prob:lex_milp}
\everymath{\displaystyle}
\begin{array}{lll}
	\max_{\eta,x_n,\epsilon} & \left[ \eta \ ; \epsilon \right]\\
        \text{s.t.}         & \eta \leq \displaystyle \sum_{a_n \in A_n} U_n^{a_{-n}, a_n} x_n^{a_n} & \forall a_{-n} \in S^+\\
	                    & x_n \in X(S^+) \cap X(S^-;\epsilon)\\
                            & \eta \in \mathbb{R}\\
                            & \epsilon \geq 0 \\
                            & x_n \in \Delta_n,
\end{array}
\end{equation}
where $\eta$ is maximised first, and $\epsilon$ second. In practice, it suffices to solve two MILPs in sequence: one in which the first objective function is maximised, and then another one in which the second objective function is maximised after imposing the first objective function to be equal to its optimal value. \qed


\end{proof}

\subsubsection{Finding an $\alpha$-Approximate Strategy}

For those cases where $f(x_n)$ does not admit a maximum, we look for a strategy $\hat x_n$ such that, for a given $\alpha > 0$, $\sup_{x_n \in \Delta_n} f(x_n) - f(\hat x_n) \leq \alpha$, i.e., for an (additively) $\alpha$-approximate strategy $\hat x_n$. Its existence is guaranteed by the following lemma.
\begin{lemma}\label{lemma}
Consider the sets $X \subseteq \mathbb{R}^n$, for some $n \in \mathbb{N}$, and $Y \subseteq \mathbb{R}$, and a function $f : X \rightarrow Y$ with $s := \sup_{x \in X} f(x)$, satisfying $s < \infty$. Then, for any $\alpha \in (0,s]$, there exists an $x \in X: s - f(x) \leq \alpha$.
\end{lemma}
\begin{proof}
By negating the conclusion, we deduce the existence of some $\alpha \in (0,s]$ such that, for every $x \in X$, $s - f(x) > \alpha$. But, then, for all $x \in X$, we have $f(x) < s - \alpha$, which implies $s = \sup_{x \in X} f(x) \leq s - \alpha < s$: a contradiction. \qed
\end{proof}

After running the algorithm we outlined in the proof of Theorem~\ref{thm:sup} to compute the value of the supremum, an $\alpha$-approximate strategy $\hat x_n$ can be computed, {\em a posteriori}, thanks to the following result:
\begin{theorem}\label{thm:apx}
Assume that $f(x_n)$ does not admit a maximum over $\Delta_n$ and that, according to the formula in Theorem~\ref{thm:sup}, $s:= \sup_{x_n \in \Delta_n} f(x_n)$ is attained at some outcome configuration $(S^+,S^-)$. Then, an $\alpha$-approximate strategy~$\hat x_n$ can be computed for any $\alpha > 0$ in, at most, exponential time by solving the following MILP:
\begin{equation}\label{prob:approx_sup}
\everymath{\displaystyle}
\begin{array}{llr}
\max_{\substack{\epsilon,x_n}} \quad & \epsilon\\
\text{{\em s.t.}} \quad & \displaystyle \sum_{a_n \in A_n} U_n^{a_{-n}, a_n} x_n^{a_n}  \geq s - \alpha & \quad \forall a_{-n} \in S^+\\
& x_n \in X(S^+) \cap X(S^-;\epsilon)\\
& \epsilon \geq 0\\
& x_n \in \Delta_n.
\end{array}
\end{equation}
\end{theorem}
\begin{proof}
Let $x_n^* \in X(S^+ \cap S^-)$ be the strategy where the supremum is attained according to the formula in Theorem~\ref{thm:sup}, namely, where $\psi(x_n^*,S^+) = \max_{\substack{x_n \in X(S^+)\cap \overline{X(S^-)}}} \psi(x_n;S^+) = s$. 
%
%
Problem~\eqref{prob:approx_sup} calls for a solution $x_n$ of value at least $s - \alpha$ (thus, for an $\alpha$-approximate strategy) belonging to $X(S^+) \cap X(S^-;\epsilon)$ with $\epsilon$ as large as possible, whose existence is guaranteed by Lemma~\ref{lemma}. Let $(\hat x_n, \hat \epsilon)$ be an optimal solution to Problem~\eqref{prob:approx_sup}. If $\hat \epsilon > 0$, $\hat x_n \in X(S^+) \cap X(S^-)$ (rather than $\hat x_n \in X(S^+) \cap \overline{X(S^-)}$). Thus, $f(x_n)$ is continuous at $x_n = \hat x_n$, implying $\psi(x_n;S^+) = f(x_n)$. Therefore, by playing $\hat x_n$, the leader achieves a utility of, at least, $s-\alpha$. \qed
\end{proof}

\subsubsection{Outline of the Explicit Enumeration Algorithm}

The complete enumerative algorithm is detailed in Algorithm~\ref{alg:ex_enum}. In the pseudocode, \textsf{CheckEmptyness}$(S^+,S^-)$ is a subroutine which looks for a value of $\epsilon \geq 0$ which is optimal for Problem~\eqref{prob:eps}, while \textsf{Solve-lex-MILP}$(S^+,S^-)$ is another subroutine which solves Problem~\eqref{prob:lex_milp}. Due to the lexicographic nature of the algorithm, $f(x_n)$ admits a maximum if and only the algorithm returns a solution with $best.\epsilon^* > 0$, whereas, if $best.\epsilon^*=0$, $x_n^*$ is just a strategy where $\sup_{x_n \in \Delta_n} f(x_n)$ is attained (in the sense of Theorem~\ref{thm:sup}). In the latter case, 
%
an $\alpha$-approximate strategy is found by invoking the procedure \textsf{Solve-MILP-approx}$(best.S^+,best.S^-,best\_value)$, which solves Problem~\eqref{prob:approx_sup}
on the outcome configuration $(best.S^+,best.S^-)$ on which the supremum has been found.




\begin{algorithm}[!htp]
	\caption{Explicit Enumeration}
	\label{alg:ex_enum}
	\begin{algorithmic}[1]
		\Function{Explicit Enumeration}{}
		\State $best \gets nil$
		\State $best\_val \gets -\infty$
		\ForAll{$S^+ \in A_F$}
		\State $ S^- \gets A_F \setminus S^+$
		\State $(\epsilon,\cdot) \gets \textsf{CheckEmptyness}(S^+,S^-)$ \Comment{Solve MILP Problem~\eqref{prob:eps}}
		\If{$\epsilon > 0$}
		\State $(\eta,\epsilon^*,x_n^*) \gets \textsf{Solve-lex-MILP}(S^+,S^-)$ \Comment{Solve lex-MILP Problem~\eqref{prob:lex_milp}}
		\If{$\eta > best\_val$}
		\State $best \gets (S^+,S^-,x_n^*,\epsilon^*)$
		\State $best\_val \gets \eta$
		\EndIf
		\EndIf
		\EndFor
                \If{$best.\epsilon^* = 0$}
                  \State $\hat x_n \gets best.x_n$
                \Else
                  \State $\hat x_n \gets \textsf{Solve-MILP-approx}(best.S^+,best.S^-,best\_val)$ \hspace{-2cm}\Comment{Solve MILP Problem~\eqref{prob:approx_sup}}
                \EndIf
		\State \Return $best\_val$, $best.x_n^*$, $\hat x_n$
		\EndFunction
	\end{algorithmic}
\end{algorithm}

\subsection{Branch-and-Bound Algorithm}\label{sub_sec:bb_algorithm}

As it is clear, computing $\sup_{x_n \in \Delta_n} f(x_n)$ with the enumerative algorithm can be impractical for any game of interesting size, as it requires the explicit enumeration of all the outcome configurations of a game---many of which will, incidentally, yield empty regions $X(S^+)\cap X(S^-)$. A more efficient algorithm, albeit one still running in exponential time in the worst-case, can be designed by relying on a branch-and-bound scheme.

\subsubsection{Computing $\sup_{x_n \in \Delta_n} f(x_n)$}

Rather than defining $S^- = A_F \setminus S^+ $, assume now $S^- \subseteq A_F \setminus S^+$. In this case, we call the corresponding pair $(S^+,S^-)$ a {\em relaxed outcome configuration}.

Starting from any followers' action profile $a_{-n} \in A_F$ with $X(a_{-n}) \neq \emptyset$, the algorithm constructs and explores, through a sequence of branching operations, two search trees, whose nodes correspond to relaxed outcome configurations. One tree accounts for the case where $a_{-n}$ is an NE and contains the relaxed outcome configuration $(S^+,S^-) = (\{ a_{-n} \},\emptyset)$ as root node. The other tree accounts for the case where $a_{-n}$ is not an NE, featuring as root node the relaxed outcome configuration $(S^+,S^-) = (\emptyset,\{ a_{-n} \})$.

If $S^- \subset A_F \setminus S^+$ (which can often be the case when relaxed outcome configurations are adopted), solving $\max_{x_n \in X(S^+) \cap \overline{X(S^-)}} \psi(x_n;S^+)$ might not give a strategy $x_n$ for which the only pure NEs in the followers' game it induces are those in $S^+$, even if $x_n \in X(S^+) \cap X(S^-)$ (rather than $x_n \in X(S^+) \cap \overline{X(S^-)}$). This is because, due to $S^+ \cup S^- \subset A_F $, there might be another action profile, say $a_{-n}' \in A_F \setminus (S^+ \cup S^-)$, providing the leader with a utility strictly smaller than that corresponding to all the action profiles in~$S^+$. Since, if this is the case, the followers would respond to $x_n$ by playing $a_{-n}'$ rather than any of the action profiles in $S^+$, $\max_{x_n \in X(S^+) \cap \overline{X(S^-)}} \psi(x_n;S^+)$ could be, in general, strictly larger than $\sup_{x_n \in \Delta_n} f(x_n)$, thus not being a valid candidate for the computation of the latter.


In order to detect whether one such $a_{-n}'$ exists, it suffices to carry out a {\em feasibility check} (on $x_n$) by looking for, in the followers' game, a pure NE different from those in $S^-$ (which may become NEs on $\bd(X(S^+) \cap X(S^-)$) which minimises the leader's utility---this can be done by inspection in $O(m^{n-1})$.
If the feasibility check returns some $a_{-n}' \notin S^+$, the branch-and-bound tree is expanded by performing a \emph{branching} operation. Two nodes are introduced: a {\em left node} with $(S_L^+,S_L^-)$ where $S_L^+ = S^+ \cup \{a_{-n}'\}$ and $S_L^- = S^-$ (which accounts for the case where $a_{-n}'$ is a pure NE), and a {\em right node} with $(S_R^+,S_R^-)$ where $S_R^+ = S^+$ and $S_R^- = S^- \cup \{a_{-n}'\}$ (which accounts for the case where $a_{-n}'$ is not a pure NE).
If, differently, $a_{-n}' \in S^+$, then $\psi(x_n;S^+)$ represents a valid candidate for the computation of $\sup_{x_n \in \Delta_n} f(x_n)$ and, thus, no further branching is needed (and $(S^+,S^-)$ is a leaf node).
%


The bounding aspect of the algorithm is a consequence of the following proposition:
\begin{proposition}\label{prop:bound}
Solving $\max_{x_n \in X(S^+) \cap \overline{X(S^-)}} \psi(x_n;S^+)$ for some relaxed outcome configuration $(S^+,S^-)$ gives an \emph{upper bound} on the leader's utility under the assumption that all followers' action profiles in $S^+$ constitute an NE and those in $S^-$ do not.
\end{proposition}
\begin{proof}
Due to $(S^+,S^-)$ being a relaxed outcome configuration, there could be outcomes not in $S^+$ which are NEs for some $x_n \in X(S^+) \cap \overline{X(S^-)}$. Due to $\psi(x_n;S^+)$ being defined as $\min_{a_{-n} \in S^+} \sum_{a_n \in A_n} U^{a_{-n}, a_n} x_n^{a_n}$, ignoring any such NE at any $x_n \in X(S^+) \cap \overline{X(S^-)}$ can only result in the $\min$ operator running on fewer outcomes $a_{-n}$, thus overestimating
$\psi(x_n;S^+)$ and, ultimately, $f(x_n)$. The claim, thus, follows. \qed
\end{proof}
As a consequence of Proposition~\ref{prop:bound}, optimal values obtained when computing the value of $\max_{x_n \in X(S^+) \cap \overline{X(S^-)}} \psi(x_n;S^+)$ throughout the search tree can be used as bounds as in a standard branch-and-bound method.

Since $\max_{x_n \in X(S^+) \cap \overline{X(S^-)}} \psi(x_n;S^+)$ is not well-defined for nodes where $S^+ = \emptyset$, for them we solve, rather than an instance of Problem~\eqref{prob:lex_milp},
a restriction of the optimistic problem (see Section~\ref{sec:problem}) with constraints imposing that all followers' action profiles in $S^-$ are not NEs. We employ the following formulation, which we introduce directly for the lexicographic case:

\begin{subequations}\label{prob:opt-with-S-}
\small
\begin{align}
\label{prob:opt-with-S-1}  \max_{y, x_n, \epsilon}
  & \; \left[\sum_{a \in A}  U_n^{a_{-n},a_n} y^{a_{-n}} x_n^{a_n}; \epsilon\right] \\
\label{prob:opt-with-S-2}  \text{s.t.}
  & \sum_{a_{-n} \in A_F} y^{a_{-n}} =1\\
\nonumber  & y^{a_{-n}} \sum_{a_n \in A_n} (U_p^{a_{-n},a_n} - U_p^{a_{-n}',a_n}) x_n^{a_n}  \geq 0 & \forall p \in F, a_{-n}\in A_F, a_p' \in A_p \setminus \{a_p\}\\
\label{prob:opt-with-S-3} 
  & & \hspace{-10cm}\text{with } a_{-n}' = a_1, \dots, a_{p-1}, a_p', a_{p+1}, \dots, a_{n-1} \hspace{-0cm}\\
\label{prob:opt-with-S-5}
  & y^{a_{-n}} \in \{0,1\} & \forall a_{-n} \in A_F\\
\label{prob: opt-with-S-6}
  & x_n \in \Delta_n\\
\label{prob: opt-with-S-7}
  & x_n \in X(S^-;\epsilon).
\end{align}
\end{subequations}
The problem can be turned into a lex-MILP by linearizing each bilinear product $y^{a_{-n}} x_n^{a_n}$ by means of McCormick's envelope and by restating Constraint~\eqref{prob: opt-with-S-7} as done in the MILP Constraints~\eqref{cons:milp}.

\subsubsection{Finding an $\alpha$-Approximate Strategy}\label{subsub:alphabnb}

Notice that, in the context of the branch-and-bound algorithm, an $\alpha$-approximate strategy $\hat x_n$ cannot be found by just relying on the {\em a posteriori} procedure outlined in Theorem~\ref{thm:apx}.
This is because, when $(S^+,S^-)$ is a relaxed outcome configuration, there might be an action profile $a_{-n}' \in A_F \setminus (S^+ \cup S^-)$ (i.e., one not accounted for in the relaxed outcome configuration) which not only is a NE in the followers' game induced by $\hat x_n$, but which also provides the leader with a utility strictly smaller than $\psi(\hat x_n; S^+)$. If this is the case, the strategy~$\hat x_n$ found with
the procedure of Theorem~\ref{thm:apx} may return a utility arbitrarily smaller than the supremum $s$ and, in particular, smaller than $s - \alpha$.

To cope with this shortcoming and establish whether such an $a_{-n}'$ exists, we first compute $\hat x_n$ according to the {\em a posteriori} procedure of Theorem~\ref{thm:apx} and, then, perform a feasibility check.
If we obtain an action profile $a_{-n}' \in S^+$, $\hat x_n$ is then an $\alpha$-approximate strategy and the algorithm halts.
If, differently, we obtain some $a_{-n}' \notin S^+$ for which the leader obtains a utility strictly smaller than $\psi(\hat x_n; S^+)$, we carry out a branching operation, creating a left and a right child node in which $a_{-n}'$ is added to, respectively, $S^+$ or $S^-$. This procedure is then applied on both nodes, recursively, until a strategy $\hat x_n$ for which the feasibility check returns
an action profile in $S^+$ is found. Such a strategy is, by construction, $\alpha$-approximate.

Observe that, due to the correctness of the algorithm for the computation of the supremum, there cannot be at $x_n^*$ an NE $a_{-n}'$ worse than the worst-case one in $S^+$. If a new outcome $a_{-n}'$ becomes the worst-case NE at $\hat x_n$, due to the fact that it is not a worst-case NE at $x_n^*$, there must be a strategy $\tilde x_n$ which is a convex combination of $x_n^*$ and $\hat x_n$ where either $a_{-n}'$ is not an NE or, if it is, it yields a leader's utility not worse than that obtained with the worst-case NE in $S^+$. An $\alpha$-approximate strategy is thus guaranteed to be found on the segment joining $\tilde x_n$ and $x_n^*$ by applying Lemma~\ref{lemma} with $X$ equal to that segment. Thus, the algorithm is guaranteed to converge.

\subsubsection{Outline of the Branch-and-Bound Algorithm}

The complete outline of the branch-and-bound algorithm is detailed in Algorithm~\ref{alg:bnb}. $\mathcal{F}$ is the frontier of the two search trees, containing all nodes which have yet to be explored. \textsf{Initialize}$()$ is a subprocedure which creates the root nodes of the two search trees, while \textsf{pick}$()$ extracts from $\mathcal{F}$ the next node to be explored. \textsf{FeasibilityCheck}$(x_n,S^-)$ performs the feasibility check operation for the leader's strategy $x_n$, looking for the worst-case pure NE in the game induced by $x_n$ and ignoring any outcome in $S^-$. \textsf{CreateNode}$(S^+,S^-)$ (detailed in Algorithm~\ref{alg:add_node}) adds a new node to $\mathcal{F}$, also computing its upper bound and the corresponding values of $x_n$ and $\epsilon$. More specifically, \textsf{CreateNode}$(S^+,S^-)$ performs the same operations of a generic step of the enumerative procedure in Algorithm~\ref{alg:ex_enum} for a given $S^+$ and $S^-$, with the only difference that, here, we invoke the subprocedure \textsf{Solve-lex-MILP-Opt}$(S^+,S^-)$ whenever $S^+=\emptyset$, by which Problem~\eqref{prob:opt-with-S-} is solved, while we invoke \textsf{Solve-lex-MILP}$(S^+,S^-)$, which solves Problem~\eqref{prob:lex_milp}, if $S^+ \neq \emptyset$. In the last part of the algorithm, $\textsf{Solve-MILP-approx}(best.S^+,best.S^-,best\_val)$ attempts to compute an $\alpha$-approximate strategy as done in Algorithm~\ref{alg:ex_enum}. In case the feasibility check fails for it, we resort to calling the procedure $\textsf{Branch-and-Bound-approx}(best.S^+,best.S^-,best.x_n^*)$ which runs a second branch-and-bound method, as described in Subsection~\ref{subsub:alphabnb}, until an $\alpha$-approximate solution is found.

\begin{algorithm}[!tp]
	\caption{Branch-and-Bound}
	\label{alg:bnb}
	\begin{algorithmic}[1]
		\Function{Branch-and-Bound}{}
		\State $best \gets nil, \quad lb \gets -\infty, \quad ub \gets \infty$
		\State $\mathcal{F} \gets \textsf{Initialize()}$
		\While{$\mathcal{F} \neq \emptyset$}
		\State $ node \gets \mathcal{F}.\textsf{pick}() $
		\If{$ node.ub > lb $}
		\State $ a_{-n} \gets \textsf{FeasibilityCheck}(node.x_n^*,node.S^-) $
		\If{$ a_{-n} \in node.S^+ $}
		\State $ best \gets (node.S^+, node.S^-, node.x_n^*, node.\epsilon^*) $
		\State $ lb \gets node.ub $
		\Else 
		\State $ S^+_L = node.S^+ \cup \{ a_{-n} \} $
		\State $ \mathcal{F} \gets \mathcal{F} + \textsf{CreateNode}(S^+_L, node.S^-) $
		\State $ S^-_R = node.S^- \cup \{ a_{-n} \} $
		\State $ \mathcal{F} \gets \mathcal{F} + \textsf{CreateNode}(node.S^+, S^-_R) $
		\EndIf
		\State $ \displaystyle ub \gets \max_{node \in \mathcal{F}} \left\{node.ub\right\} $
		\EndIf
		\EndWhile
                \If{$best.\epsilon^* = 0$}
                  \State $\hat x_n \gets best.x_n^*$
                \Else
                  \State $\hat x_n \gets \textsf{Solve-MILP-approx}(best.S^+,best.S^-,best\_val)$ \hspace{-2cm}\Comment{Solve MILP Problem~\eqref{prob:approx_sup}}
                  \State $a_{-n}' \gets \textsf{FeasibilityCheck}(\hat x_n,best.S^-)$
                  \If {$a_{-n}' \notin best.S^+$}
                    \State $\hat x_n \gets \textsf{Branch-and-Bound-approx}(best.S^+,best.S^-,best.x_n^*)$
                  \EndIf
                \EndIf
		\State \Return $ub$, $best.x_n^*, \hat x_n$
		\EndFunction
	\end{algorithmic}
\end{algorithm}

\begin{algorithm}[!htp]
	\caption{CreateNode}
	\label{alg:add_node}
	\begin{algorithmic}[1]
		\Function{CreateNode}{$ S^+, S^- $}
		\State $ (\epsilon,\cdot) \gets \textsf{CheckEmptyness}(S^+,S^-) $ \Comment{Solve MILP Problem~\eqref{prob:eps}}
		\If{$ \epsilon > 0 $}
		\State $ node \gets \textsf{EmptyNode}() $
		\State $ node.S^+ \gets S^+ $
		\State $ node.S^- \gets S^- $
		\If{$ S^+ = \emptyset $}
		\State $ (\eta,\epsilon^*,x_n^*) \gets \textsf{Solve-lex-MILP-Opt}(S^+,S^-) $ \Comment{Solve lex-MILP Problem~\eqref{prob:opt-with-S-}}
		\Else 
		\State $ (\eta,\epsilon^*,x_n^*) \gets \textsf{Solve-lex-MILP}(S^+,S^-) $ \Comment{Solve lex-MILP Problem~\eqref{prob:lex_milp}}
		\EndIf
		\State $ node.ub \gets \eta $
		\State $ node.x_n^* \gets x_n^* $
		\State $ node.\epsilon^* \gets \epsilon^* $
		\State \Return $ node $
		\EndIf
		\State \Return $ \emptyset $
		\EndFunction
	\end{algorithmic}
\end{algorithm}


\section{Experimental Evaluation}\label{sec:experiments}

We carry out, in this section, an experimental evaluation of the equilibrium-finding algorithms introduced in the previous sections. In particular, we compare three methods:
\begin{itemize}
\item {\em QCQP}: the QCQP Formulation~\eqref{prob:reform}, which we solve with the state-of-the-art spatial-branch-and-bound code BARON~\cite{sahinidis:baron:14.3.1}; note that, as indicated in~\cite{sahinidis:baron:14.3.1}, global optimality cannot be guaranteed by BARON if the feasible region of the problem at hand is not bounded, which is the case of Formulation~\eqref{prob:reform}. Hence, solutions obtained with {\em QCQP} are, in the general case, feasible but not necessarily optimal.
\item {\em MILP}: the MILP Formulation~(\ref{prob:reformMILP}), with dual variables artificially bounded by $M$, which we solve with the state-of-the-art MILP solver Gurobi~$7.0.2$; experiments with different values of $M$ are reported.
\item {\em BnB-sup}: the {\em ad hoc} branch-and-bound algorithm we proposed, described in Subsection~\ref{sub_sec:bb_algorithm} and better detailed in Algorithm~\ref{alg:bnb}, which we run to compute $\sup_{x_n \in \Delta_n} f(x_n)$, i.e., the supremum of the leader's utility. The algorithm is coded in Python~$2.7$. The different MILP subproblems that are encountered during its execution are solved with Gurobi~7.0.2.
\item {\em BnB-$\alpha$}: the {\em ad hoc} branch-and-bound algorithm we proposed, run to find, if there is no $x_n \in \Delta_n$ at which the value of the supremum is attained, an $\alpha$-approximate strategy. Results obtained with different values of $\alpha$ are illustrated.
\end{itemize}


We conduct our experiments on a testbed of normal-form game instances built with GAMUT, a widely adopted suite of game instance generators~\cite{gamut}. All the game instances we used are of \texttt{RandomGame} class, with their payoffs independently drawn from a uniform distribution with values in the range $[1,100]$.
The testbed contains games with $n = 3, 4, 5$ players (i.e., with $2, 3, 4$ followers). We generate instances with $m \in \{4,6,\ldots,20,25,\ldots,70\}$ actions when $n = 3$, and $m \in \{3,4,\ldots,15\}$ actions when $n = 4, 5$. In order to obtain statistically more robust results, the testbed includes $30$ different instances for each pair of $n$ and $m$.


Throughout the experiments and for each algorithm, we collect the following figures for each game, which we then average over all the 30 game instances in the testbed with the same values of $n$ and $m$:
\begin{itemize}
\item \emph{Time}: computing time, in seconds and up to the time limit, needed to solve the game (i.e., to compute an equilibrium).
\item \emph{LB}: the lower-bound  corresponding to the value of the best feasible solution the algorithm managed to find before halting either due to convergence or due to an elapsed time limit; by playing the strategy $x_n$ encoded in this solution, the leader is guaranteed to obtain a utility equal to at least {\em LB}; in the average, this value is only considered for instances where a feasible solution is found.
\item \emph{Gap}: the additive gap of the returned solution, measured as UB - LB, where UB is the upper-bound returned by the algorithm on the value of an optimal solution to the equilibrium-finding search problem. Note that, when solving the two restricted formulations, i.e., the QCQP Formulation~\eqref{prob:reform} and the MILP Formulation~\eqref{prob:reformMILP}, {\em Gap} corresponds to the gap ``internal'' to the solution method, thus being, in general, not valid for the original, unrestricted problem. This is not the case for BnB-sup and BnB-$\alpha$, for which {\em Gap} is a correct additive estimate of the difference between the best found LB and the value of the supremum (overestimated by UB).
\end{itemize}
We also report, for each value of $n$ and $m$, the following two figures:
\begin{itemize}
\item {\em Opt}: the percentage of instances solved to optimality. The figure is only reported for BnB since, as previously explained, optimality cannot be guaranteed for the solutions to the two formulations.
\item {\em Feas}: the percentage of instances for which a feasible solution has been found. We report the figure for the two mathematical programming formulations as an alternative to {\em Opt}.
\end{itemize}

The experiments are run on a UNIX machine with a total of 32 cores working at 2.3 GHz, equipped with 128 GB of RAM. All the computations are carried out on a single thread, with a time limit of 3600 seconds per instance. 

\subsection{Experimental Results with Two Followers}

We report, first, the results obtained on games with two followers (i.e., with $n=3$), which are summarized in Table~\ref{tab:tab1}. In particular, the table compares:
\begin{itemize} 
\item QCQP;
\item MILP with three values of $M$, namely, $M=10,100,1000$;
\item BnB-sup;
\item BnB-$\alpha$, with three values of $\alpha$, namely, $\alpha=0.1,1,10$.
\end{itemize}

\begin{table}[h!]
\rotatebox{90}{
\setlength{\tabcolsep}{1.1pt}
\begin{tabular}{r|rrrr|rrrr|rrrr|rrrr|rrrr|rH|rH|rH}
\multicolumn{24}{c}{} \\
\multicolumn{24}{c}{} \\
\multicolumn{24}{c}{} \\
\multicolumn{24}{c}{} \\
\multicolumn{24}{c}{} \\
\multicolumn{24}{c}{} \\
\multicolumn{24}{c}{} \\
\multicolumn{24}{c}{} \\
\multicolumn{24}{c}{} \\
\multicolumn{24}{c}{} \\
 & \multicolumn{4}{c|}{} & \multicolumn{4}{c|}{MILP}  & \multicolumn{4}{c|}{MILP} & \multicolumn{4}{c|}{MILP} & \multicolumn{4}{c|}{} & \multicolumn{2}{c|}{BnB-$\alpha$}  & \multicolumn{2}{c|}{BnB-$\alpha$} & \multicolumn{2}{c}{BnB-$\alpha$}\\ 
 & \multicolumn{4}{c|}{QCQP} & \multicolumn{4}{c|}{$M=10$}  & \multicolumn{4}{c|}{$M=100$} & \multicolumn{4}{c|}{$M=1000$} & \multicolumn{4}{c|}{BnB-sup} & \multicolumn{2}{c|}{$\alpha=0.1$}  & \multicolumn{2}{c|}{$\alpha=1$} & \multicolumn{2}{c}{$\alpha=10$}\\ 
\hline
$m$ & Time & LB & Gap & Fea & Time & LB & Gap & Fea & Time & LB & Gap & Fea & Time & LB & Gap & Fea & Time & LB & Gap & Opt & Time & LBe & Time & LBe & Time & LBe \\ 
\hline
4 & 3600 & 81.3 & 18.7 & 100 & 2 & 83.5 & 0.0 & 100 & 1 & 85.8 & 0.0 & 100 & 1 & 85.3 & 0.0 & 100 & 1 & 85.7 & 0.0 & 100 & 1 & 85.6 & 1 & 84.7 & 0 & 75.7 \\ 
6 & 3600 & 80.4 & 19.6 & 100 & 761 & 90.2 & 0.1 & 100 & 137 & 91.7 & 0.0 & 100 & 173 & 91.8 & 0.0 & 100 & 2 & 91.9 & 0.0 & 100 & 3 & 91.8 & 2 & 90.9 & 2 & 81.9 \\ 
8 & 3600 & 70.6 & 29.4 & 100 & 1788 & 92.9 & 1.1 & 100 & 1419 & 93.8 & 0.9 & 100 & 1760 & 93.9 & 1.1 & 100 & 5 & 94.5 & 0.0 & 100 & 9 & 94.4 & 9 & 93.5 & 9 & 84.5 \\ 
10 & 3600 & 67.2 & 32.8 & 100 & 2672 & 90.3 & 6.7 & 97 & 2161 & 95.4 & 1.7 & 100 & 2116 & 95.2 & 1.9 & 100 & 7 & 96.7 & 0.0 & 100 & 17 & 96.6 & 16 & 95.7 & 17 & 86.7 \\ 
12 & 3600 & 63.3 & 36.7 & 100 & 3456 & 84.1 & 13.0 & 100 & 3184 & 89.6 & 7.7 & 100 & 3117 & 86.7 & 10.4 & 100 & 15 & 96.8 & 0.0 & 100 & 39 & 96.7 & 39 & 95.8 & 32 & 86.8 \\ 
14 & 3600 & 57.3 & 42.7 & 97 & 3600 & 64.1 & 35.9 & 80 & 3585 & 68.3 & 30.9 & 100 & 3591 & 66.0 & 33.1 & 100 & 20 & 97.9 & 0.0 & 100 & 72 & 97.8 & 79 & 96.9 & 78 & 87.9 \\ 
16 & 3600 & 45.2 & 54.8 & 77 & 3600 & 34.1 & 65.9 & 50 & 3600 & 61.7 & 38.3 & 93 & 3600 & 59.7 & 40.3 & 100 & 53 & 97.9 & 0.0 & 100 & 226 & 97.8 & 248 & 96.9 & 230 & 87.9 \\ 
18 & 3243 & 58.3 & 41.7 & 50 & 3600 & 32.9 & 67.1 & 57 & 3600 & 53.2 & 46.8 & 83 & 3600 & 54.5 & 45.5 & 100 & 160 & 98.3 & 0.0 & 100 & 449 & 98.2 & 432 & 97.3 & 488 & 88.3 \\ 
20 & -- & -- & -- & -- & 3600 & 32.0 & 68.0 & 73 & 3600 & 41.5 & 58.5 & 93 & 3600 & 43.5 & 56.5 & 100 & 222 & 98.6 & 0.0 & 100 & 1048 & 98.5 & 1056 & 97.6 & 1075 & 88.6 \\ 
25 & -- & -- & -- & -- & 3600 & 33.0 & 67.0 & 73 & 3600 & 33.0 & 67.0 & 73 & 3600 & 33.4 & 66.6 & 73 & 777 & 99.0 & 0.0 & 100 & 3271 & 98.9 & 3112 & 98.0 & 3146 & 89.0 \\ 
30 & -- & -- & -- & -- & 3600 & 36.1 & 63.9 & 97 & 3600 & 36.1 & 63.9 & 97 & 3600 & 36.0 & 64.0 & 97 & 2687 & 95.5 & 3.8 & 47 & -- & -- & -- & -- & -- & -- \\ 
\end{tabular}
}
\caption{Experimental results for games with $n=3$ players. The figures are averaged over games with the same values of $m$.}
\label{tab:tab1}
\end{table}


As the table shows, the QCQP formulation can be solved (not to global optimality, as previously mentioned) only for instances with, at most, $m = 18$ actions, due to BARON running out of memory on larger games. Further experimental results with a larger value of $m$ are, thus, not reported. While, with $m\leq 18$, feasible solutions are found, on average, in 91\% of the cases, their quality is quite poor, as indicated by an additive gap equal to, on average, 34.6. The running times are also extremely large, with the time limit being reached on each instance, even those with $m=4$, with the sole exception of those with $m=18$, on which the solver halted prematurely due to memory issues.

Empirically, the MILP formulation performs much better than the QCQP one, allowing us to tackle instances with up to $m=30$ actions per player.
The best choice of $M$, out of the three that we have tested, seems to be $M=10$, for which we obtain, on average, LBs of 68.2 and gaps of 28.7, with a computing time slightly smaller than 2600 seconds. Although the percentage of instances where a feasible solution has been found is slightly larger with $M=1000$ (97\% as opposed to 94\%), LBs and gap become slightly worse with $M=1000$, possibly due to the fact that, as it is well-known, MILP solvers are typically quite sensitive to the magnitude of ``big M'' coefficients which, if too large, are likely to lead to large condition numbers, resulting in an algorithm which is prone to numerical issues.

As to BnB-sup, the table clearly shows that this method substantially outperforms the two mathematical programming formulations, being capable of finding not just feasible solutions, but {\em optimal} ones for {\em every} game instance with $m \leq 25$, while solving to optimality 47\% of the instances with $m=30$. We register, on average, a computing time of 359 seconds, which further reduces to 126 if we only consider the instances with $m \leq 25$ (which are all solved to optimality). Interestingly, BnB-sup shows that the supremum of the leader's utility is very large on the normal-form random games in our testbed, being equal, on average on the instances with $m \leq 25$ for which the supremum is computed exactly, to 96.

As to BnB-$\alpha$, we observe that the time taken by the method to find an $\alpha$-approximate strategy is, in essence, unaffected by the value of $\alpha$, thus allowing for the computation of $\alpha$-approximate strategy extremely close, in value, to the supremum,  without requiring a too large computational effort. We remark that, in the experiments, BnB-$\alpha$ is run only on instances with $m \leq 25$ as, in its implementation, BnB-$\alpha$ requires a relaxed outcome configuration on which the value of the supremum has been attained to compute an $\alpha$-approximate strategy. As such, experiments on games where the supremum has not been computed exactly would not give a correct solution. Note that one could, nevertheless, easily modify BnB-$\alpha$ so to look for an $\alpha$-approximate solution at each leaf node, rather than after BnB-sup has halted. This way, if the method halts due to the time limit being met, one would still obtain a solution which is $\alpha$-approximate w.r.t. the value of the best (LB) estimation of the supremum that has been found within the time limit.

\begin{table}[htbp]
\caption{Results obtained with BnB-sup for games with $n=3$ players and $35 \leq m \leq 70$.}
\begin{center}
\begin{tabular}{r|rrrr}
 & \multicolumn{4}{c}{BnB-sup}  \\ 
\hline
$m$ & Time & LB & Gap & \%Opt \\ 
\hline
35 & 3573 & 78.8 & 20.9 & 3 \\ 
40 & 3560 & 63.1 & 36.8 & 0 \\ 
45 & 3600 & 50.2 & 49.8 & 0 \\ 
50 & 3600 & 48.8 & 51.2 & 0 \\ 
55 & 3600 & 52.3 & 47.7 & 0 \\ 
60 & 3600 & 48.9 & 51.1 & 0 \\ 
65 & 3600 & 49.2 & 50.8 & 0 \\ 
70 & 3600 & 49.1 & 50.9 & 0 \\ 
\end{tabular}
\end{center}
\label{tab:tab2}
\end{table}

Table~\ref{tab:tab2} reports further results obtained with BnB-sup for games with up to $m=70$ actions per player. As the table shows, while some optimal solutions can still be found for $m=35$, optimality is lost for any game instance with $m \geq 40$. Interestingly, though, BnB-sup still manages to find feasible solutions for instances up to $m=70$. On average, the method obtains solutions with an average LB of 55.1 and an average additive gap of 44.9. Under the conservative assumption that, in the testbed, games with $35 \leq m \leq 70$ admit suprema of value close to 100 (which is empirically true when $m \leq 30$), BnB-sup provides, on average, solutions that are less than 50\% off of optimal ones.


\subsection{Experimental Results with More Followers and Final Observations}

Results obtained with BnB-sup with more than two followers (i.e., with $n=4,5$) are reported in Table~\ref{tab:tab3} for $m \leq 14$.

\begin{table}[htbp]
\begin{center}
\caption{Results obtained with BnB-sup for games with $n=3,4,5$ players and $4 \leq m \leq 70$.}
\label{tab:tab3}
\begin{tabular}{r|rrr|rrr|rrr}
 & \multicolumn{3}{c|}{BnB-sup} & \multicolumn{3}{c|}{BnB-sup} & \multicolumn{3}{c}{BnB-sup}\\ 
 & \multicolumn{3}{c|}{$n=3$} & \multicolumn{3}{c|}{$n=4$} & \multicolumn{3}{c}{$n=5$}\\
\hline
$m$ & Time & Gap & Opt & Time & Gap & Opt & Time & Gap & Opt \\ 
4 & 0 & 0.0 & 100 & 3 & 0.0 & 100 & 8 & 0.0 & 100 \\ 
6 & 2 & 0.0 & 100 & 17 & 0.0 & 100 & 137 & 0.0 & 100 \\ 
8 & 5 & 0.0 & 100 & 126 & 0.0 & 100 & 2953 & 11.3 & 53 \\ 
10 & 7 & 0.1 & 100 & 955 & 0.0 & 100 & 3461 & 45.0 & 13 \\ 
12 & 15 & 0.0 & 100 & 2784 & 5.7 & 60 & 3600 & 52.9 & 0 \\ 
14 & 20 & 0.1 & 100 & 3600 & 49.9 & 0 & 3600 & 51.8 & 0 \\ 
16 & 53 & 0.0 & 100 & - & - & - & - & - & - \\ 
18 & 160 & 0.0 & 100 & - & - & - & - & - & - \\ 
20 & 222 & 0.1 & 100 & - & - & - & - & - & - \\ 
30 & 2687 & 3.8 & 47 & - & - & - & - & - & - \\ 
40 & 3560 & 36.8 & 0 & - & - & - & - & - & - \\ 
50 & 3600 & 51.2 & 0 & - & - & - & - & - & - \\ 
60 & 3600 & 51.1 & 0 & - & - & - & - & - & - \\ 
70 & 3600 & 50.9 & 0 & - & - & - & - & - & - \\ 
\end{tabular}
\end{center}
\end{table}

As the table illustrates, computing the value of the supremum of the leader's utility becomes very hard already for $m=12$ with $n=4$, for which the algorithm manages to find optimal solution in only 60\% of the cases. For $m=14$, no instance is solved to optimality within the time limit. For $n=5$, the problem becomes hard already for $m=8$, for which only 53\% of the instances are solved to optimality, whereas, for $m=12$, no instances at all are solved to optimality.


We do not report results on game instances with $n=4,5$ and $m > 14$ as such games are so large that, on them, BnB-sup incurs memory problems due to the MILP subproblems it solves.


In spite of the problem of computing a P-LFPNE being a nonconvex pessimistic bilevel program, with our branch-and-bound algorithm we can find solutions with an optimality gap $\leq 0.01$ for 3-player games with up to within $m=20$ actions (containing three payoffs matrices with 8000 entries each), which are comparable, in size, to those solved in previous works which solely tackled the problem of computing a single NE maximising the social welfare, see, e.g.,~\cite{SandholmGC05}.

\section{Conclusions}\label{sec:conclusions}


We have shown that the problem of computing a pessimistic leader-follower equilibrium with multiple followers playing pure strategies simultaneously and noncooperatively is \textsf{NP}-hard with two or more followers and inapproximable in polynomial time when the number of followers is three or more unless $\textsf{P}=\textsf{NP}$.
We have proposed an exact single-level QCQP reformulation for the problem, with a restricted version which we cast into an MILP, and an exact exponential-time algorithm (which we have then embedded in a branch-and-bound scheme) for finding
the supremum of the leader's utility and, in case there is no leader's strategy where such value is attained, also an $\alpha$-approximate strategy.

Future developments include applications to structured games (e.g., congestion games), establishing the approximability status of the problem with two followers, and the generalization to the case with both leader and followers playing mixed strategies---even though we conjecture that this problem could be much harder, probably $\Sigma_2^p$-complete.


\bibliographystyle{spbasic}      
\bibliography{paper}   




\end{document}